\documentclass[twoside,leqno,twocolumn]{article}

\usepackage[letterpaper]{geometry}

\usepackage{siamproceedings}

\usepackage[T1]{fontenc}
\usepackage{amsfonts}
\usepackage{graphicx}
\usepackage{epstopdf}
\usepackage{enumitem}
\usepackage{algorithmic}
\usepackage{todonotes}

\ifpdf
  \DeclareGraphicsExtensions{.eps,.pdf,.png,.jpg}
\else
  \DeclareGraphicsExtensions{.eps}
\fi

\newsiamremark{remark}{Remark}
\crefname{remark}{Remark}{Remarks}
\newsiamremark{hypothesis}{Hypothesis}
\crefname{hypothesis}{Hypothesis}{Hypotheses}
\newsiamremark{observation}{Observation}
\crefname{observation}{Observation}{Observations}
\crefname{appendix}{Appendix}{Appendices}
\Crefname{appendix}{Appendix}{Appendices}
\newsiamthm{claim}{Claim}

\AtBeginEnvironment{definition}{\itshape}
\AtBeginEnvironment{observation}{\itshape}

\usepackage{amsopn}

\usepackage{etoolbox}
\usepackage{wrapfig}
\usepackage{nccmath}
\usepackage{nicematrix}
\usepackage[linewidth=1pt]{mdframed}
\usepackage{mathtools}

\usepackage{pgfplots}
\usepgfplotslibrary{groupplots}
\usepackage{tikz}

\usepackage{subcaption}
\usepackage{stfloats}
\usepackage{xspace}

\definecolor{my-dark-red}{RGB}{183, 28, 28}
\definecolor{my-red}{RGB}{244,67,54}
\definecolor{my-pink}{RGB}{233,30,99}
\definecolor{my-purple}{RGB}{156,39,176}
\definecolor{my-deep-purple}{RGB}{103,58,183}
\definecolor{my-indigo}{RGB}{63,81,181}
\definecolor{my-blue}{RGB}{33,150,243}
\definecolor{my-light-blue}{RGB}{3,169,244}
\definecolor{my-cyan}{RGB}{0,188,212}
\definecolor{my-teal}{RGB}{0,150,136}
\definecolor{my-green}{RGB}{76,175,80}
\definecolor{my-light-green}{RGB}{139,195,74}
\definecolor{my-lime}{RGB}{205,220,57}
\definecolor{my-yellow}{RGB}{255,235,59}
\definecolor{my-amber}{RGB}{255,193,7}
\definecolor{my-orange}{RGB}{255,152,0}
\definecolor{my-deep-orange}{RGB}{255,87,34}
\definecolor{my-brown}{RGB}{121,85,72} 
\definecolor{my-grey}{RGB}{120,120,120}
\definecolor{my-blue-grey}{RGB}{96,125,139}
\definecolor{my-lipics-grey}{rgb}{0.6,0.6,0.61}

\tikzset{
  empty/.style={mark=x,every mark/.append style={mark size=0.0pt}},
  3_aprx/.style={mark=diamond*,my-green,semithick,draw=black,every mark/.append style={mark size=3.2pt}},
  exact/.style={mark=square*,my-blue,semithick,draw=black,every mark/.append style={mark size=2.5pt}},
  exact_smpl/.style={mark=*,my-yellow,semithick,draw=black,every mark/.append style={mark size=2.5pt}},
  lpf/.style={mark=triangle*,my-red,semithick,draw=black,every mark/.append style={mark size=3.2pt}},
  kkp2/.style={mark=pentagon*,my-purple,semithick,draw=black,every mark/.append style={mark size=2.8pt}},
  pfp/.style={mark=pentagon*,my-purple,semithick,draw=black,every mark/.append style={mark size=2.5pt}},
  alz4/.style={mark=halfsquare*,my-deep-orange,semithick,draw=black,every mark/.append style={mark size=2.5pt}},
  alz8/.style={mark=halfsquare right*,my-deep-orange,semithick,draw=black,every mark/.append style={mark size=2.5pt}},
  alz12/.style={mark=halfsquare left*,my-deep-orange,semithick,draw=black,every mark/.append style={mark size=2.5pt}},
  lz77topk/.style={mark=halfdiamond*,my-deep-orange,semithick,draw=black,every mark/.append style={mark size=3.2pt}},
  zstd/.style={mark=square*,my-yellow,semithick,draw=black,every mark/.append style={mark size=2.5pt}},
  gzip/.style={mark=triangle*,my-red,semithick,draw=black,every mark/.append style={mark size=3.2pt}},
  bzip2/.style={mark=*,my-indigo,semithick,draw=black,every mark/.append style={mark size=2.85pt}},
  xz/.style={mark=asterisk,my-light-blue,thick},
  7z/.style={mark=Mercedes star,my-light-green,thick},
  lz4/.style={mark=+,my-deep-orange,thick},
  ssszip_bsc/.style={mark=diamond*,my-green,semithick,draw=black,every mark/.append style={mark size=3.2pt}},
  alz_8/.style={mark=pentagon*,my-pink,semithick,draw=black,every mark/.append style={mark size=2.8pt}},
  alz_6/.style={mark=pentagon*,my-pink,semithick,draw=black,every mark/.append style={mark size=2.8pt}},
  bsc_2047/.style={mark=x,my-grey,semithick,draw=black,every mark/.append style={mark size=2.8pt}},
  marks/.style={
    every mark/.append style={mark size=3pt},
  },
  legendmarks/.style={
    marks,
    every plot/.append style={only marks}
  },
  chart/.style={
    legend label/.style={font={\footnotesize},anchor=west,align=left},
    legend box/.style={rectangle, draw, minimum size=5pt},
    axis/.style={black,semithick,->},
    axis label/.style={anchor=east,font={\footnotesize}},
  }
}

\pgfplotsset{
  title style={yshift=-5pt, font=\bfseries, inner sep=0pt},
  major grid style={thin,dotted,color=my-blue-grey},
  minor grid style={thin,dotted,color=my-grey!45},
  ymajorgrids,
  yminorgrids,
  xmajorgrids,
  xminorgrids,
  every tick label/.append style={font=\footnotesize},
  every axis label/.append style={font=\footnotesize},
  legend cell align=left,
  legend pos=outer north east,
  scaled x ticks=false,
  scaled y ticks=false,
  max space between ticks=30,
  minor tick num=2,
  every axis/.append style={
    line width=0.5pt,
    tick style={
      line cap=round,
      thin,
      major tick length=4pt,
      minor tick length=2pt,
    },
  },
  xticklabel style={
    /pgf/number format/fixed,
    /pgf/number format/precision=3
  },
  yticklabel style={
    /pgf/number format/fixed,
    /pgf/number format/precision=3
  },
  throughput_style/.style={
    scale only axis,
    ylabel near ticks,          
    xlabel near ticks,
    align=center,
    title style={at={(0.5,1.05)}},
    max space between ticks=20,
  },
  legend/.style={
    hide axis,
    scale only axis,
    width=1pt,
    height=1pt,
    xmin=0,xmax=1,
    ymin=0,ymax=1,
    legend style={
        font=\footnotesize,
        anchor=center,
        at={(0.5,0.5)},
        /tikz/every even column/.append style={column sep=0.1cm}
    },
  },
}

\newcommand{\LPF}{\textsc{LPF}\xspace}
\newcommand{\LCP}{\mathrm{LCP}}
\newcommand{\LCE}{\mathrm{LCE}}
\newcommand{\SA}{\mathrm{SA}}
\newcommand{\PA}{\mathrm{PA}}
\newcommand{\ISA}{\mathrm{SA}^{-1}}
\newcommand{\IPA}{\mathrm{PA}^{-1}}
\newcommand{\PSV}{\mathrm{PSV}}
\newcommand{\NSV}{\mathrm{NSV}}
\newcommand{\ioOroR}{\textsc{IO-OROR}\xspace}
\newcommand{\swOroR}{\textsc{SW-OROR}\xspace}
\newcommand{\ssszip}{\texttt{ssszip}\xspace}
\newcommand{\Oh}{{\mathcal O}}
\newcommand{\polylog}{\mathrm{polylog}\,}

\newcommand{\siv}{\mathrm{siv}}
\newcommand{\piv}{\mathrm{piv}}
\newcommand{\SIV}{\mathrm{SIV}}

\newcommand{\RMQ}{\mathrm{RMQ}}
\newcommand{\per}{\mathrm{per}}

\newcommand{\Hread}{H_{\mathrm{read}}}
\newcommand{\Hwrite}{H_{\mathrm{write}}}
\newcommand{\Mpeak}{M_{\mathrm{peak}}}
\newcommand{\Mcur}{M_{\mathrm{cur}}}
\newcommand{\MH}{M_H}

\newcommand{\ceil}[1]{\lceil #1 \rceil}
\newcommand{\floor}[1]{\lfloor #1 \rfloor}

\def\drawannotate(#1,#2,#3){
    \path (axis cs:#1,#2)-- +(10pt,4pt) node[rotate=60,scale=.3,pos=.25] {\phantom{v/}} node[font=\tiny] {#3};
}

\newcommand\annotatenear[1]{\drawannotatenear(#1)}
\def\drawannotatenear(#1,#2,#3,#4,#5){%
    \path (axis cs:#1,#2) -- +(#3,#4)
        node[font=\scriptsize,inner sep=0.9pt,text=black,fill=white,fill opacity=0.9,text opacity=1]{#5};%
}

\def\drawannotatelead(#1,#2,#3,#4,#5){%
    \draw[my-grey,line width=0.2pt] (axis cs:#1,#2) -- +(#3,#4)
        node[font=\scriptsize,inner sep=0.9pt,text=black,fill=white,fill opacity=0.9,text opacity=1]{#5};%
}

\def\drawannotateoffset(#1,#2,#3,#4,#5){
    \path (axis cs:#1,#2)-- +(#3,#4) node[rotate=60,scale=.3,pos=.25] {\phantom{v/}} node[font=\tiny] {#5};
}

\makeatletter
\newcommand{\ELSEIFONLY}[1]{%
  \STATE \algorithmicelse\ \algorithmicif\ #1\ \algorithmicthen
  \begin{ALC@g}%
}
\newcommand{\ENDELSEIFONLY}{%
  \end{ALC@g}%
}
\makeatother

\let\eps\epsilon
\newcommand{\absolute}[1]{\left\lvert#1\right\rvert}
\let\emptystring\varepsilon

\newif\iffullversion
\fullversiontrue
\newcommand{\apx}[1]{\iffullversion\cref{#1}\else the full version\fi}
\newcommand{\Apx}[1]{\iffullversion\Cref{#1}\else The full version\fi}

\theoremstyle{plain}
\theoremheaderfont{\normalfont\scshape}
\theorembodyfont{\normalfont\itshape}
\theoremseparator{.}
\theoremsymbol{}
\newtheorem*{restatedtheorem}{Theorem}

\hypersetup{
  pdftitle={Practical and Space-Efficient LZ77 and LZ Pre-Compression via String Synchronizing Sets},
  pdfauthor={Jonas Ellert and Lukas Nalbach},
  pdfsubject={Full version of the paper accepted at ALENEX 2027},
  pdfkeywords={LZ77, data compression, string synchronizing sets, compressed indexing}
}

\begin{document}

\newcommand\relatedversion{}
\iffullversion\else
\renewcommand\relatedversion{\thanks{The full version of the paper can be accessed at \protect\url{https://arxiv.org/abs/2609.30193}}}
\fi

\title{\Large Practical and Space-Efficient LZ77 and LZ Pre-Compression via String Synchronizing Sets \relatedversion}
    \author{Jonas Ellert\thanks{University of Southern Denmark, Odense, Denmark (\email{ellert.jonas@gmail.com}).
    Partially supported by the ERCIM ``Alain Bensoussan'' Fellowship Programme (at CWI Amsterdam) and by grant ANR-20-CE48-0001 (at DI/ENS Paris).}
    \and Lukas Nalbach\thanks{TU Dortmund University, Germany (\email{lukas.nalbach@tu-dortmund.de}).}}
\date{}

\maketitle

\begin{abstract}
\looseness=-1
The Lempel--Ziv (LZ77) factorization decomposes a text into the least possible number $z$ of phrases that each refer to an earlier occurrence.
It is this phrase count, rather than the encoded size, that governs the size of LZ-based compressed indexes, and computing a factorization with few phrases is a time and space bottleneck in their construction.
In practice, computing LZ77 quickly typically requires the suffix array.

In theory, the exact LZ77 factorization, and a 3-approximation of it, can be computed more space-efficiently without the full suffix array [SPIRE 2023].
These algorithms have remained unimplemented, because two of their components resist a direct implementation: a lookup table that degenerates to patterns of length at most two for realistic inputs, and an orthogonal range reporting data structure that is impractical.
We replace both, fine-tune every remaining stage, and obtain the first practical implementation, which runs in space close to the text rather than to the suffix array.
On one thread, our 3-approximation factorizes $12$--$19\times$ faster than the classical LPF algorithm while using $14\times$ less memory; on 32 threads, even our exact algorithm is $1.4$--$2.9\times$ faster than parallel LPF, at $9\times$ less memory.
In practice the approximation ratio stays far below 3.
As a side result, passing only the perfect phrases of our 3-approximation to a downstream compressor yields a precompressor that is on par with the state of the art \cite{alz} in compression ratio, and better in memory consumption and parallel throughput.
\end{abstract}

\medskip

\noindent\textit{Artifact:} \url{https://doi.org/10.5281/zenodo.22892236}

\fancyfoot[C]{\small\thepage}
\pagenumbering{arabic}
\setcounter{page}{1}

\section{Introduction}\label{sec:intro}
The exact \emph{Lempel--Ziv factorization (LZ77)} decomposes a length-$n$ text $T$ into $z$ phrases $f_1 f_2 \ldots f_z = T$, where each phrase $f_i$ is either the leftmost occurrence of a single character (a \emph{literal phrase}) or the longest prefix of $f_i \ldots f_z$ that has an earlier occurrence (a \emph{referencing phrase} or \emph{reference}) \cite{LZ76,ZL77,Storer}%
\footnote{The original factorization by Lempel and Ziv works slightly differently, and was actually introduced in 1976 \cite{LZ76} (with a practical extension in 1977 \cite{ZL77}).
Nevertheless, as is common in the overwhelming majority of the literature, we say LZ77 to refer to the Storer--Szymanski version of the factorization \cite{Storer}.}.
In an \emph{LZ-like factorization}, we drop the requirement that references need to be of maximal length.
The LZ77 factorization minimizes the number of phrases among all LZ-like factorizations \cite{LZ76,Storer}.
An LZ-like factorization is called an \emph{$\alpha$-approximate LZ factorization} if it consists of at most $\alpha z$ phrases.
Since each reference has an earlier occurrence, it can be encoded in constant space (e.g., as a position-length pair), allowing the text to be decompressed in a left-to-right fashion.
Hence, using the exact LZ77 factorization, one can encode the text in $\Oh(z)$ space (or in $\Oh(\alpha z)$ space using an $\alpha$-approximate factorization), significantly compressing repetitive texts.

Many practical compressors are built on LZ-like factorizations (e.g., zlib/DEFLATE\footnote{\url{https://www.zlib.net/}, RFC~1951.},
gzip\footnote{\url{https://www.gnu.org/software/gzip/}, RFC~1952.},
zstd\footnote{\url{https://facebook.github.io/zstd/}, RFC~8878.},
and LZ4\footnote{\url{https://lz4.org/}}).
However, unlike exact LZ77, these compressors are not too concerned with minimizing the number of phrases; instead, a big part of the achieved compression comes from a highly optimized encoding of the references.
This is appropriate for the setting these tools were designed for: the storage and transmission of comparatively small, low-redundancy files, such as individual web assets served over HTTP, where the encoded byte size and (de)compression speed matter far more than the phrase count\footnote{Both Brotli (RFC~7932) and zstd (RFC~8878) are, for instance, standardized HTTP content encodings.}.
However, in other settings, this approach has major drawbacks:
\begin{description}
    \item[Compressing repetitive collections:] Whenever the text is highly repetitive with long repeating substrings (e.g., a collection of human genomes, which are on average more than $99\%$ identical to one another \cite{1000Genomes,NHGRIvariation}), the achieved compression is inadequate (with, e.g., gzip taking more space than a naive $2z\ceil{\log_2 n}$-bit encoding of exact LZ77).
    \item[Compressed indexing:] In many applications, we not only want to store the text in compressed form, but also need to support fast pattern matching queries.
    This can be achieved with a compressed LZ index \cite{lz_index,lz_index_2,NavarroSurveyII}, which locates the primary occurrences of a pattern using $\Oh(z)$ space and derives the secondary ones from those, again in $\Oh(z)$ space \cite[Section 3]{NavarroSurveyII}.
    In this setting, a small compressed representation is of little use if it consists of many phrases: the number of phrases controls the size of the index.
    The same applies to grammar-based compression, for which the factorization is the standard starting point: from $z$ phrases one obtains a grammar of size $\Oh(z\log(n/z))$, which is within an $\Oh(\log(n/z))$ factor of the smallest grammar \cite{grammar_from_lz}.
    \item[Compressed computation:] A growing family of algorithms takes an arbitrary LZ-like factorization of $T$ -- exact or approximate, of size $z'$ -- as its \emph{only} input and computes on it in $\Oh(z' \polylog n)$ time, never decompressing $T$.
    Examples are the construction of a small grammar \cite{grammar_from_lz}, of the run-length compressed BWT \cite{KempaKociumaka20bwt}, and of an index supporting $\LCE$, internal pattern matching, $\SA$ and $\ISA$ queries \cite{KempaKociumaka23}.
    Here the phrase count bounds not only the space but also the \emph{running time} of everything built on top of the factorization, so keeping it close to $z$ pays off twice.
\end{description}

For these applications, one should try to minimize the number of phrases, which raises the question whether there is a practical algorithm for computing an LZ-like factorization that consists of close to $z$ phrases.

\subsection*{Our contributions.}
We present a practical implementation for computing LZ-like factorizations whose phrase count is close to $z$ -- or exactly $z$ -- intended as a building block for LZ-based compressed indexes, grammar compression, and compressed computation.
We build on the algorithms of \cite{lz77_sublinear}, which compute exact LZ77 in $\Oh(n/\log_\sigma n + z\log^{3+\epsilon} z)$ time and a 3-approximate factorization in $\Oh(n/\log_\sigma n)$ time, both in $\Oh(n/\log_\sigma n)$ words of space.
Contrary to many other LZ algorithms, their key components (string synchronizing sets, geometric data structures, text indexes for short patterns) admit practical implementations -- with two exceptions, which is why the algorithms have remained unimplemented.
The short-pattern lookup table of \cite[Lemma 3]{lz77_sublinear} degenerates to patterns of length $m \leq 2$ for realistic inputs, and the insertion-only orthogonal range one-reporting (\ioOroR) data structure of \cite[Theorem 1]{io_oror_ds} is impractical.
We replace both and obtain the first practical implementation.
In detail:

\begin{itemize}
      \item \textbf{Practical 3-approximation (\Cref{sec:3_apprx_opt}).}
      We replace the short-pattern lookup table with a tunable rolling Karp--Rabin hash index that maps fingerprints of multiple pattern lengths to their previous occurrence positions.
      Combined with an optimized LPF-phrase computation, this factorizes $12$--$19\times$ faster than the classical \LPF algorithm on one thread, while using $14\times$ less memory and producing only $1.40$--$2.02\,z$ phrases.

      \item \textbf{Practical exact algorithm (\Cref{sec:exact_opt}).}
      We implement \ioOroR with a simple grid data structure and decompose the instance into $\sigma$ smaller ones (\Cref{thm:decomp}), shrinking the effective grid area by up to a factor of $\sigma$.
      On 32 threads this is $1.4$--$2.9\times$ faster than the parallel \LPF algorithm at $9\times$ less memory.

      \item \textbf{Interval sampling (\Cref{sec:sampling}).}
      Apart from the \ioOroR queries, the bottleneck in the exact algorithm is the computation of sparse $\SA$- and $\PA$-intervals.
      We precompute hash tables for a set of adaptively chosen pattern lengths, allowing interval lookups in expected $\Oh(1)$ time and reducing total running time by $1.7$--$3.2\times$.

      \item \textbf{Parallelization (\Cref{sec:parallel}).}
      We parallelize every stage of both algorithms, obtaining speedups of $2.8$--$6.3\times$ for the 3-approximation and $13$--$26\times$ for the exact algorithm on 32 threads.

      \item \textbf{Practical compression tool (\Cref{sec:compression}).}
      The perfect phrases of the 3-approximation yield a practical precompressor that we call \ssszip.
      When paired with \texttt{bsc}\footnote{\url{http://libbsc.com/}}, it is on par with the state-of-the-art \texttt{alz} \cite{alz} in compression ratio, and better in memory consumption and parallel throughput.
\end{itemize}

\subsection*{Related work.}
Computing the exact factorization in $\Oh(n)$ time and $\Oh(n)$ words of working space is a classical exercise: textbook solutions simulate a left-to-right scan of the text using the suffix tree \cite{RodehPrattEven81,Gusfield97,CHL07}, while more practical approaches derive the factorization from the suffix array, e.g., by computing previous and next smaller values (PSV/NSV) over the array \cite{CrochemoreIlieSmyth08,kkp,GotoBannai13,KKP16lazy}.
However, the $\Oh(n)$ words of space may prohibitively exceed the $n \ceil{\log_2 \sigma}$ bits occupied by the text itself, assuming alphabet $[0..
\sigma)$ with ${\sigma \ll n}$.
This motivated a long line of research on more space-efficient algorithms, culminating in algorithms that use $\Oh(n \log \sigma)$ bits, i.e., space close to the text \cite{OhlebuschGog11,KKP13sea,GotoBannai14,Kosolobov15,FischerIKoeppl15,FIKS18}.
Most recently, Kempa and Kociumaka \cite{KempaKociumaka24} showed that, within this space, the exact LZ77 factorization can even be computed in slightly sublinear $\Oh((n \log \sigma) / \sqrt{\log n})$ time.
For a text in read-only memory, $\Oh(z)$ additional working space suffices to compute a $(1+\eps)$-approximation of LZ77 (i.e., an LZ-like factorization consisting of at most $(1+\eps) \cdot z$ phrases), for any constant $\eps > 0$ \cite{FGGK15}.
However, these results were aimed at improving the theoretical complexity and admit no straightforward efficient implementation.

In terms of compressed indexing, solutions based on the (run-length compressed) Burrows--Wheeler transform \cite{bwt} have matured into practical tools -- largely thanks to scalable construction via prefix-free parsing \cite{pfp,KuhnleMBGLM20,RossiOLGB22}.
Practical implementations of LZ-based indexes \cite{lz_index,lz_index_2,FerradaGHP14hybrid} remain scarce, in no small part because computing a factorization with few phrases is a time and space bottleneck in their construction \cite{kkp,KKP13sea}.

\subsection*{Roadmap.}
After the preliminaries, we recall the 3-approximation and exact algorithm of \cite{lz77_sublinear} (\Cref{sec:approx,sec:exact}) and then present our practical implementation of each (\Cref{sec:algorithmic_optimizations}), including the compression tool \ssszip.
\Cref{sec:experiments} evaluates them against the state of the art.

\section{Preliminaries}\label{sec:prelim}
A \emph{string} $T$ of length $n = \absolute{T}$ is a sequence of $n$ symbols from some alphabet $\Sigma$.
Throughout this work, we assume $\Sigma = \{1, \ldots ,\sigma\}$ with $\sigma = n^{\Oh(1)}$, implying the lexicographical order of strings in the usual way.
We analyze algorithms in the word-RAM model \cite{wordram} with word width $w = \Theta(\log n)$, and we further assume that $T$ is given in \emph{word-packed} representation: each symbol occupies $\lceil\log_2\sigma\rceil$ bits, so $T$ fits in $\Oh(n/\log_\sigma n)$ words.
For integers $i,j$, we use interval notation $[i,j] = [i, j + 1) = (i - 1, j] = (i - 1, j + 1) = \{ k \in \mathbb Z \mid i \leq k \leq j\}$.
If $1 \leq i \leq j\leq n$, then we define substring $T[i,j] = T[i]T[i+1] \ldots T[j]$, and otherwise $T[i,j] = \emptystring$, where $\emptystring$ is the empty string.
We call $T_i = T[i,n]$ the $i$-th suffix of $T$.
A string $T$ has \emph{period} $p$ if $T[i] = T[i+p]$ for all $i \in [1,|T|-p]$, and $\per(T)$ denotes the smallest such $p$.
Let $\SA[1..n]$ be the \emph{suffix array} \cite{suffix_arrays}: the unique permutation of $[1,n]$ with $T_{\SA[1]} < T_{\SA[2]} < \ldots < T_{\SA[n]}$.
We can build $\SA$ in $\Oh(n)$ time \cite{sais}.
Let $\ISA[1..n]$ with $\ISA[\SA[i]] = i$ be the \emph{inverse suffix array}.
We define the \emph{longest common extension} $\LCE(i,j)$ as the length of the longest common prefix of $T_i$ and $T_j$.
Let the \emph{LCP array} $\LCP[1..n]$ be defined by $\LCP[1] = 0$ and $\LCP[i] = \LCE(\SA[i-1],\SA[i])$ for $i \in [2, n]$.
For an array $A[1..n]$ of integers, the \emph{range minimum query} asks for the position of the minimum value in a given range $[i, j]$ in $A$, i.e.,\ $\RMQ_A(i,j) = \arg\min_{k \in [i,j]} A[k]$.
One can easily see that given $\ISA$, $\LCP$ and an $\Oh(n)$-space and $\Oh(1)$-time $\RMQ_{\LCP}$ data structure, we can answer LCE queries in $\Oh(1)$ time \cite{rmq_o1} by $\LCE(i, j) = \LCP[\RMQ_{\LCP}(\ISA[i] + 1, \ISA[j])]$, where $\ISA[i] < \ISA[j]$ holds w.l.o.g.\ (else, swap $i$ and $j$).

\paragraph{LZ77 factorization.}
Let the \emph{LZ77 factorization} \cite{LZ76} of $T$ be $f_1 f_2 \ldots f_z = T$, where each phrase $f_j$ at destination $i$ is either a \emph{literal phrase} (if $T[i]$ has no prior occurrence) or a \emph{referencing phrase} with source $s \in [1, i)$ and maximal length $l$ s.t.\ $T[s, s+l) = T[i, i+l)$.
We have $z = \Oh(n/\log_\sigma n)$.
An \emph{LZ-like factorization} relaxes maximality: referencing phrases may be shorter than the longest match.

\paragraph{LPF algorithm.}
Let $\PSV[1..n]$ and $\NSV[1..n]$ be the \emph{previous} and \emph{next smaller value} arrays w.r.t.\ $\SA$, where $\PSV[i] = \max\{j \in [0,i) : \SA[j] < \SA[i]\}$ and $\NSV[i] = \min\{j \in (i,n] : \SA[j] < \SA[i]\}$ (or 0 if undefined).
Let $\LPF[i] = \arg\max_{s \in [1, i)} \LCE(s,i)$ be the \emph{longest previous factor} of $i$, or $\bot$ if $T[i]$ has no prior occurrence.
Setting $j = \ISA[i]$, $x = \SA[\PSV[j]]$ and $y = \SA[\NSV[j]]$, any match longer than $\LCE(i, x)$ or $\LCE(i, y)$ starts after position $i$; hence, $\LPF[i] = \arg\max_{p \in \{x, y\}} \{\LCE(i, p)\}$.
Given $\SA$, $\ISA$, $\PSV$, $\NSV$, and an $\Oh(1)$-time LCE data structure (all computable in $\Oh(n)$ time and words), we extract the LZ77 factorization in one additional $\Oh(n)$-time pass \cite{CrochemoreIlieSmyth08,kkp}.

\paragraph{Karp--Rabin fingerprinting.}
Fix a prime $q = \Theta(n^c)$ with $c \geq 2$ and $q \geq 2\sigma$, and a base $b \in [\sigma, q)$.
We define the \emph{Karp--Rabin fingerprint} of $T[i,j]$ as $\varphi(i,j) = \bigl(\sum_{k=i}^{j} T[k] \cdot b^{j-k}\bigr) \bmod q$.
Equal substrings yield equal fingerprints; unequal ones collide with probability $\Oh(1/n)$ for a uniformly random $b$ \cite{karp_rabin}.

\paragraph{String synchronizing sets.}

\begin{definition}[\cite{lsss}]\label{def:sss}
      We call a set $S \subseteq [1, n-2\tau+1]$ \emph{$\tau$-synchronizing} w.r.t.\ $T$ if
      $\mathbf{(i)}$ synchronizing condition: for all $i,j$ with $T[i,i+2\tau) = T[j,j+2\tau)$, we have $i \in S \Leftrightarrow j \in S$, and
      $\mathbf{(ii)}$ density condition: for every $i \in [1, n-3\tau+2]$, $S \cap [i,i+\tau) = \emptyset$ iff $\per(T[i, i+3\tau-2]) \leq \tau/3$.
\end{definition}

The synchronizing condition ensures that equal length-$l \geq 2\tau$ substrings contain samples at identical offsets in $[0,l-2\tau]$.
Such sets of size $\Oh(n/\tau)$ can be constructed \cite{lsss}.

\paragraph{Sparse data structures.}
Let $S = [p_1, \ldots , p_{|S|}]$ with $p_1 < \ldots < p_{|S|}$ be a set of sampled text positions.
Let the \emph{sparse suffix array} $\SA_S$ be the permutation of $[1,|S|]$ sorting $T_{p_1}, \ldots , T_{p_{|S|}}$ lexicographically, and let $\ISA_S$ be its inverse.
Let the \emph{sparse prefix array} $\PA_S$ sort the same positions by their prefixes $T[1,p_i]$ in colexicographic order; let $\IPA_S$ be its inverse.

We define the \emph{sparse LCP array} $\LCP_S[1..|S|]$ by $\LCP_S[1] = 0$ and $\LCP_S[i] = \LCE(S[\SA_S[i-1]], S[\SA_S[i]])$ for $i \in [2, |S|]$.
We define the \emph{sparse PSV/NSV arrays} $\PSV_S[1..|S|]$ and $\NSV_S[1..|S|]$ analogously to $\PSV$ and $\NSV$, but w.r.t.\ $\SA_S$.
We define the \emph{sparse LPF array} $\LPF_S[1..|S|]$ by $\LPF_S[i] = \arg\max_{j \in [1, i)} \LCE(p_j, p_i)$ (or 0 if undefined).

We define the \emph{sparse suffix array interval} $\siv_S(P)$ for a pattern $P \in \Sigma^*$ as the maximal interval $[y_1, y_2] \subseteq [1,|S|]$ such that for every $k \in [y_1, y_2]$, the suffix $T_{S[\SA_S[k]]}$ starts with $P$.
We define $\piv_S(P)$ analogously using $\PA_S$.
We can compute both in $\Oh(\log|S|)$ time via binary search over $\SA_S$ and $\PA_S$ using $\Oh(1)$-time LCE queries on $T$ and the reverse of $T$, respectively.

\paragraph{Constant-time LCE queries in sublinear space.}
Recall that LCE queries can be answered in $\Oh(1)$ time and $\Oh(n)$ space using $\ISA$, $\LCP$ and an $\Oh(n)$-space and $\Oh(1)$-time $\RMQ_{\LCP}$ data structure.
Similarly, if we build an SSS with $\tau = \log_\sigma n$, then we can answer LCE queries in $\Oh(1)$ time and $\Oh(n / \tau) = \Oh(n / \log_\sigma n)$ space using $S$, $\ISA_S$, $\LCP_S$ and an $\Oh(|S|)$-space and an $\Oh(1)$-time $\RMQ_{\LCP_S}$ data structure (see \cite{lsss} for the details).

\begin{theorem}[\cite{lsss}]\label{thm:lce_sublinear}
      We can build in $\Oh(n/\log_\sigma n)$ time and space a data structure for $\Oh(1)$-time LCE queries.
\end{theorem}

\section{Algorithms}
We recall the two algorithms of \cite{lz77_sublinear} that our implementation builds on: the LZ77 3-approximation, which computes an LZ-like factorization of at most $3z$ phrases, and the exact algorithm, which computes the LZ77 factorization itself on top of it.

\subsection{LZ77 3-approximation}\label{sec:approx}
The 3-approximation proposed in \cite{lz77_sublinear} 
computes an LZ-like factorization in which phrases are either \emph{perfect} or not perfect.
In an LZ-like factorization $f_1\dots f_{z'}$, a phrase $f_i$ is perfect if and only if either $i = z'$, or $f_{i}\cdot f_{i + 1}[1]$ has no occurrence in $f_1\dots f_{i}$, i.e., the phrase is indeed of maximal length.
It is easy to see that the number of perfect phrases is at most $z$.

The algorithm first uses \Cref{thm:lce_sublinear} to compute a \emph{gapped LZ factorization} with only perfect phrases.

\begin{definition}\label{def:gapped_lz}
      We call the sequence $f_1 g_1 r_1 \ldots f_{z'} g_{z'} r_{z'} = T$ a gapped LZ factorization if
      each $f_x$ is a perfect phrase (of maximal length) starting at some sample $p_x \in S$,
      each $g_x$ is a (possibly empty) gap containing no sample from $S$, and
      each $r_x$ is a (possibly empty) referencing phrase.
      Every gapped LZ factorization satisfies $z' \leq z$ \cite{lz77_sublinear}.
\end{definition}

Starting from a $\tau$-SSS $S$ with $\tau = \lfloor\log_2 n / (8\lceil\log_2\sigma\rceil)\rfloor$, the algorithm consists of three phases.
It additionally uses a lookup table of size $\Oh(n/\log_\sigma n)$ that stores the leftmost occurrence of any length-$m \leq 3\tau$ pattern \cite{lz77_sublinear}.

\textbf{Phase 1 (LPF phrases at sampled positions).}
We scan $S$ left to right.
For each sample $p = S[i]$, we compute $\LPF_S[i]$ via $\PSV_S$, $\NSV_S$, and the $\Oh(1)$-LCE data structure of \Cref{thm:lce_sublinear}.
If $T[p, p+2\tau)$ has a previous occurrence (case 1; detected via the lookup table), then setting $j = \ISA_S[i]$ implies $\LPF_S[i] \in \{S[\SA_S[\PSV_S[j]]],\, S[\SA_S[\NSV_S[j]]]\}$ (whichever gives the longer LCE).
Otherwise (case 2), we can find $\LPF_S[i]$ using the lookup table in $\Oh(\tau)$ time.

\textbf{Phase 2 (shortening long gaps).}
Every gap $g_x$ with $|g_x| > 3\tau$ is periodic with period $p \leq \tau/3$ by the density condition of the SSS.
We append a referencing (possibly non-perfect) phrase $r_x$ covering $g_x[3\tau+1,|g_x|]$ with source $i_x + 3\tau - p$ (where $i_x$ is the gap's start), shortening $g_x$ to length $3\tau$.
We retrieve the period $p$ in $\Oh(1)$ from a precomputed table of size $\Oh(\sqrt{n} \log_\sigma n)$ \cite{lz77_sublinear}.

\textbf{Phase 3 (closing short gaps).}
Each remaining gap $g_x$ with $|g_x| \leq 3\tau$ is either \emph{referencing} (has a previous occurrence, which we close in $\Oh(1)$ time via the pattern table) or \emph{non-referencing} (a leftmost occurrence).
We replace each non-referencing gap with literal and referencing phrases from left to right:
We iteratively find its longest prefix with a previous occurrence in $\Oh(\tau)$ time via the pattern table.
Note that only the last such added phrase can be non-perfect.

There are at most $z$ perfect phrases in total and at most two non-perfect phrases per gap; hence, $z'' \leq z + 2 z' \leq 3z$.
All three phases run in $\Oh(n / \log_\sigma n)$ time and space in total \cite{lz77_sublinear}.

\begin{theorem}[\cite{lz77_sublinear}]\label{thm:3aprx}
      We can compute an LZ-like factorization with at most $3z$ phrases in $\Oh(n/\log_\sigma n)$ time and $\Oh(n/\log_\sigma n)$ words of space.
\end{theorem}

\subsection{Exact LZ77 Algorithm}\label{sec:exact}
The exact algorithm of \cite{lz77_sublinear} builds on the 3-approximation.
It first constructs a special sample set $C \subseteq [1,n]$ and lexicographically and colexicographically sorts its samples.
Computing each perfect phrase then reduces to finding sparse prefix and suffix array intervals and \ioOroR queries.

\begin{figure*}[t]
      \centering
      \includegraphics[width=0.75\linewidth]{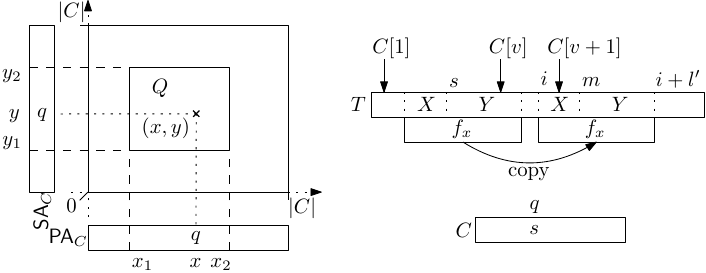}
      \caption{Illustration of a far source. $X$ is the head, and $Y$ is the tail.}
      \label{fig:far_source}
\end{figure*}

\begin{lemma}[\cite{lz77_sublinear}]\label{lem:sample_set}
      Given an LZ-like factorization with at most $3z$ phrases, we can construct in $\Oh(n / \log_\sigma n)$ time and space a sample set $C \subseteq [1,n]$ of size $\Oh(n / \log_\sigma n)$ that is \emph{$\delta$-dense} (every window $[i,i+\delta) \subseteq [1,n]$ contains a sample) and \emph{leftmost-substring-covering} (the leftmost occurrence $T[i,i+l)$ of every substring of $T$ contains a sample).
\end{lemma}
\begin{proof}
      Take $C = E \cup \Delta$, the union of the phrase endpoints $E$ of the 3-approximation and $\Delta = \{i\delta \mid i \in [1, \lfloor n/\delta \rfloor]\}$.
      Setting $\delta = \Omega(\log_\sigma n)$ yields $|C| = \Oh(z + n / \log_\sigma n) = \Oh(n / \log_\sigma n)$.
\end{proof}

\begin{definition}[\cite{lz77_sublinear}]\label{def:io_oror}
      Let $\pi$ be a permutation of $[1,N]$ and let $P = \{(i,\pi(i)) \mid i \in [1,N]\}$ be the set of valid points.
      Let \emph{insertion-only orthogonal range one-reporting} (\ioOroR) be the problem of maintaining an initially empty set $R \subseteq P$ under
      \emph{Insert}$(p)$: add $p \in P \setminus R$ to $R$; and
      \emph{Query}$(Q)$ with $Q = [x_1,x_2] \times [y_1,y_2]$: return any $p \in Q \cap R$, or report $Q \cap R = \emptyset$.
\end{definition}

Now, we can discuss the main algorithm.

\begin{theorem}\label{thm:exact_lz77}
      We can compute the exact LZ77 factorization in $\Oh(n / \log_\sigma n + z \log^{3+\epsilon} z)$ time and $\Oh(n / \log_\sigma n)$ space.
\end{theorem}
\begin{proof}
We use \Cref{thm:lce_sublinear}, \Cref{thm:3aprx} and \Cref{lem:sample_set} to obtain $C$.
Then, we build $\SA_C$, $\PA_C$, $\ISA_C$ and $\IPA_C$ using comparison-based sorting and $\Oh(1)$-time LCE queries in $\Oh(|C| \log |C|)$ time.
We describe one \ioOroR instance $R$ on $P = \{(\IPA_C[i], \ISA_C[i]) \mid i \in [1, |C|]\}$; the theoretical bound uses two instances, as explained below.

The exact factorization is then computed left to right.
Suppose we have already factorized $T[1, i)$, i.e.,\ $i$ is the destination of the next phrase to compute.
We maintain the following invariants:

\textit{Invariants}.
We maintain an index $v$, the \emph{frontier}, defined as the index such that $C[v] < i \leq C[v+1]$, that is, $v$ is the index of the last sample before the current phrase destination $i$, and we adjust $v$ after computing a phrase.
Furthermore, we maintain that $R$ contains exactly the points $P' = \{(\IPA_C[i], \ISA_C[i]) \mid i \in [1, v]\} \subseteq P$.

\textit{Computing a perfect phrase}.
Let $f$ be the perfect phrase starting at position $i$.
We find the length $l$ and a source $s$ of $f$ via two methods, and keep track of the source that yields the longest match:

\begin{enumerate}
\item \textbf{Close sources}:
      Here, we consider each close source $s \in [i-\delta,i)$ by computing $\LCE(s, i)$.
      This takes $\Oh(\delta)$ time.
      
\item \textbf{Far sources}:
      Here we try to find a phrase $T[i, i + l')$ decomposed into a \emph{head} $T[i,m]$ and a \emph{tail} $T[m, i + l')$, where $m \in [i, i + l')$ is called the \emph{split point} (see the right part of \Cref{fig:far_source}).
      For each head length $k \in [1,\delta]$, we set $m = i+k-1$ and perform an exponential search for the maximum $l'$ s.t.\ there is a $q \in [1, v]$ with $T[i,m] = T(C[q] - k, C[q]]$ and $T[m, i + l') = T[C[q], C[q] + l' - k]$, i.e.,\ we can use $C[q] - k + 1$ as a source of a length-$l'$ phrase.
      Whether this $q$ exists is equivalent to
      \begin{align*}
            &\PA_C[x_1,x_2] \cap \SA_C[y_1,y_2] \cap [1,v] \neq \emptyset \\
            \Leftrightarrow &[x_1,x_2] \times [y_1, y_2] \cap P' \neq \emptyset,
      \end{align*}
      where $[x_1,x_2] = \piv_C(T[i, m])$ and $[y_1,y_2] = \siv_C(T[m, i+l'))$.
      Hence, we can find such $q$ (if it exists) via an \ioOroR query to $R$ with rectangle $Q = [x_1, x_2] \times [y_1, y_2]$ (see the left part of \Cref{fig:far_source}).
      We obtain $q$ from the point $p = (x, y)$ returned by the query by $q = \PA_C[x]$ or $q = \SA_C[y]$.
\end{enumerate}

To show that we indeed find a perfect phrase, consider the leftmost occurrence $\lambda$ of this perfect phrase.

\textbf{Case 1} $\lambda \in [i - \delta, i)$:
In this case, we find $\lambda$ in method 1.

\textbf{Case 2.1} $\lambda < i - \delta$ and $l < \delta$:
Then there is a sample $C[q] \in [\lambda, \lambda + l)$, because $C$ is leftmost-substring-covering.

\textbf{Case 2.2} $\lambda < i - \delta$ and $l \geq \delta$:
Then there is a sample $C[q] \in [\lambda, \lambda + \delta)$ because $C$ is $\delta$-dense.

Thus, in both Cases 2.1 and 2.2, for $k = C[q] - \lambda + 1$ and $l' = l$ during method 2, we have $q \in \PA_C[x_1, x_2] \cap \SA_C[y_1, y_2] \cap [1, v]$.
Hence, the query to $R$ with the rectangle $Q = [x_1, x_2] \times [y_1, y_2]$ will report a point $(\IPA_C[u], \ISA_C[u])$, where either $u = q$, or $u \in (q, v]$ is another sample index such that $T(C[u] - k, C[u] + l - k] = T[i, i + l)$.

The theoretical bound requires keeping the leftmost-substring-covering set $A=E$ and the $\delta$-dense set $B=\Delta$ separate: the two structures of \cite{lz77_sublinear} have insertion times $\Oh(\log^{3+\epsilon} n)$ on $\Oh(z)$ samples and $\Oh(\log n)$ on $\Oh(n/\delta)$ samples, respectively.
Computing intervals and querying both instances with $\delta=\Theta(\log^2 n)$ gives the claimed total time and space by the amortized analysis in \cite{lz77_sublinear}.
\end{proof}

\section{Algorithmic Optimizations}\label{sec:algorithmic_optimizations}
In this section, we discuss our practically optimized implementations of the LZ77 3-approximation and the exact LZ77 algorithm.
Both build the SSS with $\tau = 512$; \apx{apx:parameters} explains how we chose this and all following parameters.

\subsection{3-Approximation}\label{sec:3_apprx_opt}
We first run Phase 1 and obtain all LPF phrases and the gap positions (see \Cref{sec:lpf_impl}).
We then process the gaps in a second left-to-right pass using the rolling hash index described in \Cref{sec:hash_index}.

\begin{figure}[t]
      \centering
      \includegraphics[width=0.9\linewidth]{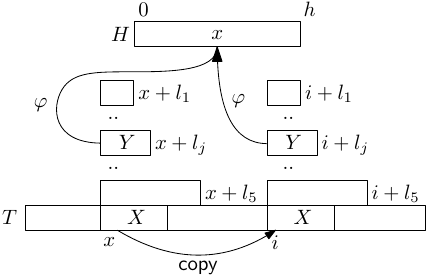}
      \caption{Illustration of how we can find an LZ phrase $X$ using the rolling hash index.}
      \label{fig:rolling_hash_index}
\end{figure}

\subsubsection{LPF-Phrase Computation}\label{sec:lpf_impl}
Here, we employ two simple practical optimizations.
First, we extend an LPF phrase to the left if it is preceded by a gap.
Second, we use the following observation to skip samples in the SSS that are covered by LPF phrases.

\begin{observation}\label{obs:skip}
      Let $f_x$ be a perfect phrase from \Cref{def:gapped_lz} starting at sample $p_x$ with source $s_x$.
      Let $p_j \in S$ be maximal such that $p_x < p_j < p_x + |f_x|$, and let $l_j = \LCE(p_j, \LPF_S[j])$.
      Then any sample $p_k \in S$ with $p_x < p_k < p_j$ yields a sparse LPF phrase ending before $p_j + l_j$.
\end{observation}

This immediately follows from the synchronizing condition of the SSS and the definition of the sparse LPF.
Hence, we can skip each such $p_k$, as it contributes nothing beyond the LPF from $p_j$.

\subsubsection{Gap Factorization}\label{sec:hash_index}
For each gap region we only perform Phase 3, because Phase 2 (closing long, periodic gaps) changed neither the running time nor the approximation ratio in practice.

The pattern lookup table from \cite[Lemma 3]{lz77_sublinear} requires patterns of length $l \leq \lfloor\log_2 n/((2+\varepsilon)\lceil\log_2\sigma\rceil)\rfloor$.
For practical values such as $\sigma = 256$ and $n = 2^{32}$, this gives $l \leq 2$, making it useless.
Thus, we replace it with a \emph{rolling hash index} $H$.

$H$ is a flat array of $h$ entries, each storing one text position.
Because we will use fingerprints of strings to index $H$ (which requires computing $\varphi \bmod h$), we want $h$ to be a power of 2, because then, $\bmod h$ is equivalent to $\& (h-1)$, where $\&$ denotes ``bitwise and''.
We therefore give $H$ a memory budget of $\MH = \min(2^{30}, \max(2^{20}, n/10, g/3, \Mpeak - \Mcur))$ bytes and set $h$ to the power of two that results in the size of $H$ being closest to $\MH$.
Here, $g$ is the total length of the gaps, and $\Mpeak - \Mcur$ is the memory that the preceding phases have already released below their own peak, which $H$ can occupy without increasing the peak memory consumption.

We select a set $L = \{l_1, l_2, l_3, l_4, l_5\}$ of pattern lengths based on a gap phrase length guess $\rho = \min(\bar{g}, \bar{l}, 2^{10-7g/n})$, where $\bar{g}$ is the mean gap length and $\bar{l}$ is the mean LPF phrase length.
\Cref{tab:lengths} gives the mapping from $\rho$ to $L$.

\begin{table}[t]
\centering
\setlength{\tabcolsep}{4pt}
\renewcommand{\arraystretch}{1.0}
\begin{tabular}{cc|cc}
\hline
$\rho$ range & $L$ & $\rho$ range & $L$ \\
\hline
$(0,6]$   & $\{2,3,4,5,6\}$   & $(32,64]$      & $\{2,4,7,12,28\}$ \\
$(6,8]$   & $\{2,3,4,6,8\}$   & $(64,128]$     & $\{2,4,8,16,36\}$ \\
$(8,12]$  & $\{2,3,4,8,12\}$  & $(128,256]$    & $\{2,5,10,20,42\}$ \\
$(12,16]$ & $\{2,4,6,9,16\}$  & $(256,1024]$   & $\{2,6,12,24,48\}$ \\
$(16,32]$ & $\{2,4,6,10,20\}$ & $(1024,\infty)$ & $\{2,8,16,32,64\}$ \\
\hline
\end{tabular}
\caption{The set $L$ of pattern lengths maintained by the rolling hash index, as a function of $\rho$; the table continues in the right half.}
\label{tab:lengths}
\end{table}

During the left-to-right scan at position $p$, for each length $l_j \in L$ we maintain the Karp--Rabin rolling fingerprint $\varphi(p, p+l_j-1)$.
As we do not store any fingerprints, space is not a limitation, so we can afford the large Mersenne prime $q = 2^{107}-1$ with 128-bit registers.
At a position $p$ inside a gap, we compute for each $l_j \in L$ an index $i = \varphi(p, p+l_j-1) \bmod h$ in $H$ and retrieve a source position $s_j = H[i]$.
Then, we use the source maximizing $\LCE(s_j, p)$ for the next phrase.
Finally, we write $H[i] \leftarrow p$.

When trying to find a phrase starting at $T[i]$, we possibly find a previous occurrence $x$ of some $Y = T[i, i + l_j)$ in $H$ at position $\varphi(i, i + l_j - 1)$, because when we last rolled over $Y$, we set $H[\varphi(i, i + l_j - 1)] \gets x$ (see \Cref{fig:rolling_hash_index}).
Finally, we extend the match further if possible to obtain an even longer phrase $X$.

We employ a further optimization:
We do not roll over LPF phrases (except for the first position in each phrase), because this only increases running time and does not improve compression.

Note that in contrast to the theoretical description of the 3-approximation, we compute an LPF phrase at a sample $p_i$ even if $T[p_i, p_i + 2\tau)$ has no previous occurrence.
The resulting phrase may thus not be perfect.
Therefore, we allow a gap phrase to extend past the following LPF phrase, which improves compression in practice.

\subsection{Exact LZ77 Algorithm}\label{sec:exact_opt}
In this section, we discuss our optimizations for the exact LZ77 algorithm.
\cite{master_thesis} describes this implementation, including the exact algorithm with sparse interval sampling (\Cref{sec:sampling}), in full detail.

\subsubsection{Sample-Set Construction}\label{sec:build_C}
In practice, we build $C$ (see \Cref{lem:sample_set}) from the approximate factorization in a single left-to-right pass.
We set $\delta = \min(\lfloor n/z'' \rfloor, 256)$, where $z''$ is the number of phrases of the approximate factorization.
Rather than adding all multiples of $\delta$ globally, we only insert extra samples between consecutive phrase endpoints in $E$ when their distance exceeds $\delta$.
This reduces the size of $C$ by $\approx 30\%$ compared to the generic construction.

\subsubsection{Orthogonal Range One-Reporting}\label{sec:oror}
Recall that during the computation of a perfect phrase starting at position $i$, we use two methods (see \Cref{thm:exact_lz77}).
Method 1 considers close sources $s \in [i - \delta, i)$.
Method 2 finds far sources $s < i - \delta$, which involves \ioOroR queries (see \Cref{def:io_oror}).

If the sequence of inserted points is fixed in an instance of \ioOroR, we call it an \ioOroR instance with \emph{fixed insertion order}.
In the exact LZ77 algorithm, the sample set $C$ sets $N = |C|$ and $\pi(\IPA_C[i]) = \ISA_C[i]$, and points are inserted in the order $(\IPA_C[1],\ISA_C[1]), \ldots , (\IPA_C[|C|],\ISA_C[|C|])$ as the left-to-right scan reveals new samples, so the insertion order is fixed.
In this special case, we can reduce \ioOroR to a static variant of the problem:

\begin{definition}\label{def:sw_oror}
      Let $\pi$ be a permutation of $[1,N]$.
      The \emph{static weighted orthogonal range one-reporting} (\swOroR) problem is, given a set $P = \{(i,\pi(i)) \mid i \in [1,N]\}$ of $N$ points and a weight function $w$, to answer \emph{Query}$(Q, v)$ with $Q = [x_1,x_2] \times [y_1,y_2]$:
      return any $p \in Q \cap P$ with $w(p) \leq v$, or report that no such $p$ exists.
\end{definition}

\begin{remark}\label{rem:reduction}
      Any \ioOroR instance with fixed insertion order $(i_1,\pi(i_1)), \ldots ,(i_M,\pi(i_M))$ can be reduced to \swOroR by setting $w(i_j, \pi(i_j)) = j$ for $j \in [1,M]$.
      A \emph{Query}$(Q)$ with $Q = [x_1,x_2] \times [y_1,y_2]$ to the \ioOroR instance after $v$ insertions becomes \emph{Query}$(Q, v)$ to \swOroR.
\end{remark}

Hence, we can instead use a data structure for \swOroR in the exact LZ77 algorithm.
Furthermore, this enables us to achieve two things:
First, we can omit method 1 (close sources) in the proof of \Cref{thm:exact_lz77} by dynamically adjusting $v$ for each $k$ such that $C[v] = \max (C \cap [1, i+k-1))$ during method 2 (far sources), where $v = 0$ if $C \cap [1,i+k-1) = \emptyset$.
Second, being able to arbitrarily choose the query weight enables us to parallelize the exact LZ77 algorithm by letting each thread factorize its own part of $T$.

\paragraph{Practical Data Structure.}\label{sec:grid}
We implement \swOroR using a simple grid data structure.
We divide the overall point range $[1,N]^2$ into a $G \times G$ grid of cells with side length $D = 2^{14}$, i.e.,\ $G = \lceil N / D \rceil$.
Each point $p = (x,y) \in P$ belongs to cell $c(p) = (\lceil x/D \rceil - 1)\, G + \lceil y/D \rceil$; since $A$ and $I$ are $1$-based, we round both coordinates up so that this \emph{cell index} is an integer in $[1, G^2]$.
We store all points in an array $A[1..N]$ sorted by their cell indexes, with ties broken by weight.
Furthermore, we store the array $I[1..G^2]$, where $I[j] = |\{p \in P \mid c(p) < j\}| + 1$ is the beginning of the range in $A$ of the points with cell index $j$.
$I$ can thus be used during a query to find all points $A[I[i], I[i+1])$ in cell $i$.
Overall, our data structure occupies $\Oh(N + (N/D)^2)$ words.

Answering a query involves iterating over all cells cut by the query rectangle $Q = [x_1,x_2] \times [y_1,y_2]$.
First, we check the cells fully contained in $Q$.
If the first point in one of those cells is heavier than the query weight, we can skip this cell, because the points in each cell are sorted by their weight.
Then, we check the boundary cells, which are not fully contained in $Q$.
Here we have to check the weight and the position of each point.

$Q$ fully contains $|Q| / D^2$ cells.
If we find a point in one such cell, we can either directly return it or skip the cell; hence, we spend $\Oh(|Q| / D^2)$ time for the fully contained cells.
$Q$ cuts $< 4G = \Oh(N / D)$ boundary cells.
For each point in a boundary cell, we have to check whether it lies in $Q$.
However, since the points have distinct $x$- and $y$-coordinates, there are $< 4D$ such points in total; hence, we need overall $\Oh(N / D + D)$ time for the boundary cells.
In total, we need $\Oh(|Q| / D^2 + N / D + D)$ time.

\paragraph{Character Decomposition.}\label{sec:decomp}
The query rectangles in our exact algorithm are constrained by the character at the sample position, which allows the global \swOroR instance to be decomposed into $\sigma$ independent sub-instances with a smaller overall point range (see \Cref{thm:decomp}).

\begin{figure}[t]
      \centering
      \includegraphics[width=0.9\linewidth]{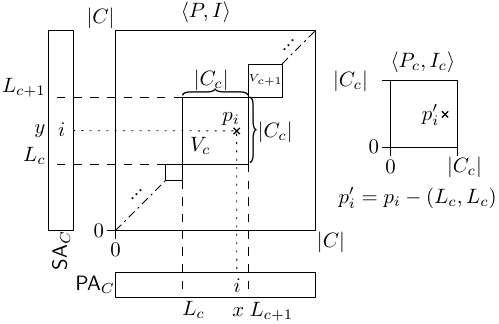}
      \caption{Illustration of the character decomposition described in \Cref{thm:decomp}.}
      \label{fig:decomposed_range}
\end{figure}

\begin{theorem}\label{thm:decomp}
      For $c \in [1, \sigma]$, define
      $C_c = \{p \in C \mid T[p] = c\}$,
      $L_c = \sum_{i=1}^{c-1}|C_i|$,
      $V_c = (L_c, L_{c+1}]^2$ and
      $P_c = P \cap V_c$ with $L_{\sigma + 1} = |C|$.
      Then, during the exact LZ77 algorithm (\Cref{thm:exact_lz77}), each query rectangle $Q = [x_1,x_2] \times [y_1,y_2]$ with $[x_1,x_2] = \piv_C(T[i, m])$ and $[y_1,y_2] = \siv_C(T[m, i + l'))$ is contained in $V_{T[m]}$.
      Therefore, the global \ioOroR instance $P$ can be reduced to the $\sigma$ independent \ioOroR instances $P_1, \ldots , P_\sigma$ with fixed insertion order.
\end{theorem}
\begin{proof}
Clearly, $P_1, \ldots , P_\sigma$ are valid instances of \ioOroR.
We simulate $P$ using $P_1, \ldots , P_\sigma$ as follows:
Each Insert$(p)$ with $p = (\IPA_C[i], \ISA_C[i]) = (x, y)$ is delegated to instance $P_c$ with $c = T[C[i]]$ as Insert$(p')$ with $p' = (x, y) - (L_c, L_c)$, because this implies $p \in P_c$.
Similarly, each Query$(Q, w')$ with $Q = [x_1,x_2] \times [y_1,y_2]$, where $[x_1,x_2] = \piv_C(T[i, m])$ and $[y_1,y_2] = \siv_C(T[m, i + l'))$, is delegated to instance $P_c$ with $c = T[m]$, query rectangle $Q' = [x_1',x_2'] \times [y_1', y_2'] = [x_1 - L_c, x_2 - L_c] \times [y_1 - L_c, y_2 - L_c]$ and weight $w'$.
A returned point $p' = (x', y') \in P_c$ is translated back to its corresponding point $p = (x', y') + (L_c, L_c) \in P$.
\end{proof}

\Cref{fig:decomposed_range} illustrates the decomposition.
Assuming a uniform character distribution, i.e.,\ $|C_c| = |C|/\sigma$ for all $c \in \Sigma$, the total point range of the $\sigma$ sub-instances is
\[
  \sum_{c \in \Sigma} |C_c|^2
  = \sigma \left(\frac{|C|}{\sigma}\right)^{\!2}
  = \frac{|C|^2}{\sigma}
\]
instead of $|C|^2$, reducing the effective grid area by a factor of $\sigma$.
This enables us to either decrease the cell width $D$ without increasing the size of the data structure, or reduce the size of the data structure while leaving $D$ unchanged.

\paragraph{Small-Range Optimization.}
For query rectangles that are narrow in at least one dimension, i.e.,\ $\min(x_2-x_1, y_2-y_1) \leq \gamma$ for a threshold $\gamma = 4096$, we bypass the query to $R$ and answer the query by a linear scan over two precomputed arrays $\Pi[1..|C|]$ and $\Psi[1..|C|]$, where $\Pi[x] = \ISA_C[\PA_C[x]]$ and $\Psi[y] = \IPA_C[\SA_C[y]]$.
Conceptually, $\Pi[x]$ maps a prefix-rank $x$ to the suffix-rank of the same sample.
Symmetrically, $\Psi[y]$ maps a suffix-rank $y$ to the prefix-rank of the same sample.
To answer a query $\PA_C[x_1,x_2] \cap \SA_C[y_1,y_2] \cap [1,v] \neq \emptyset$, we scan whichever of $[x_1,x_2]$ and $[y_1,y_2]$ is smaller: scanning $x \in [x_1,x_2]$ and checking $\Pi[x] \in [y_1,y_2]$ and $\PA_C[x] \leq v$, or scanning $y \in [y_1,y_2]$ and checking $\Psi[y] \in [x_1,x_2]$ and $\SA_C[y] \leq v$.
This causes only one random memory access and is thus faster than the grid data structure from \Cref{sec:grid} for short ranges.

\subsubsection{Sparse Prefix and Suffix Array Interval Sampling}\label{sec:sampling}
The bottleneck in the exact algorithm is the binary searches for sparse suffix and prefix array intervals w.r.t.\ $C$, each requiring $\Oh(\log|C|)$ LCE queries.
We reduce this cost by sampling all intervals for all patterns of specific lengths.
More precisely, we choose a set $\Lambda$ of \emph{sampled pattern lengths}, and for each $\lambda \in \Lambda$ we precompute a hash table $H(\SIV^\lambda_C)$ holding $\siv_C(T[p, p+\lambda))$ for every sample $p \in C$.
We use the fact that in the exact algorithm, each queried $P$ is a substring $T[i, i + l)$ of $T$, so its fingerprint is at hand.
$H(\SIV^\lambda_C)$ stores these intervals, hashing $[b,e]$ to $\varphi(p_b, p_b + \lambda - 1)$ with $p_b = C[\SA_C[b]]$, i.e.,\ to the fingerprint of the substring that $[b,e]$ itself identifies.
A lookup of $\siv_C(T[i, i+\lambda))$ thus scans the bucket of $\varphi(i, i+\lambda-1)$ and returns the first $[b,e]$ in it with $\LCE(p_b, i) \geq \lambda$, or $\emptyset$ if none exists: one $\Oh(1)$-time LCE query acts as the \emph{equality check}.
As the expected bucket size is $\Oh(1)$, the additional time to compute $\siv_C(P)$ is $\Oh(1)$ in expectation instead of $\Oh(\log|C|)$.
For \emph{unsampled} lengths $\lambda' \notin \Lambda$, the time to find $\siv_C(T[i, i + \lambda'))$ is not reduced to $\Oh(1)$.
However, the stored intervals at the two adjacent sampled lengths serve as tighter initial brackets for the binary search (see \Cref{rem:interpolate}).

\begin{remark}\label{rem:interpolate}
      Let $l_1,l_2,l_3$ be pattern lengths with $1 \leq l_1 \leq l_2 \leq l_3$.
      Let $[b_1, e_1] = \siv_C(T[i, i + l_1))$, $[b_2, e_2] = \siv_C(T[i, i + l_2))$ and $[b_3, e_3] = \siv_C(T[i, i + l_3))$.
Then, $b_1 \leq b_2 \leq b_3$ and $e_3 \leq e_2 \leq e_1$.
\end{remark}

If $l_1$ and $l_3$ are sampled, but not $l_2$, we compute $[b_1, e_1]$ and $[b_3, e_3]$ in $\Oh(1)$ time via hash table lookups, and then find $b_2$ and $e_2$ via binary searches over $[b_1, b_3]$ and $[e_3, e_1]$, respectively.

\paragraph{Computing Fingerprints of Substrings.}
Recall that to compute the hash tables and to find $\siv_C(T[x, y))$, we need $\varphi(x, y)$.
To quickly find $\varphi(x, y)$ for arbitrary $x$ and $y$, we store the fingerprint of each $s$th prefix of $T$, using $s = 16$ and $q = 2^{31}-1$.
This yields the following time/space trade-off:

\begin{theorem}\label{thm:fingerprints}
Let $s \geq 1$ be an integer parameter.
Then there is a data structure with size $\Oh(n / s + \sqrt{n})$ and construction time $\Oh(n)$, s.t.\ we can compute any $\varphi(i, j)$ in $\Oh(s)$ time.
\end{theorem}
\begin{proof}
We store the array $F[0..\floor{n/s}]$ with $F[i] = \varphi(1, is)$, so that $F[0] = \varphi(1, 0) = 0$ is the fingerprint of the empty prefix.

To compute an arbitrary $\varphi(i, j)$, we first look up $\varphi(1, \floor{(i-1)/s}s) = F[\floor{(i-1)/s}]$ and $\varphi(1, \floor{j/s}s) = F[\floor{j/s}]$, and extend both to $\varphi(1, i-1)$ and $\varphi(1, j)$ in $\Oh(s)$ time.
We then compute
\[
      \varphi(i, j) = \bigl(\varphi(1, j) - b^{\,j-i+1}\,\varphi(1, i-1)\bigr) \bmod q .
\]
It remains to evaluate $b^{\,j-i+1} \bmod q$ in $\Oh(1)$ time.
To this end, we precompute two arrays $X[0..\ceil{\sqrt{n}}]$ and $Y[0..\ceil{\sqrt{n}}]$ with $X[a] = b^a \bmod q$ and $Y[a] = b^{a\ceil{\sqrt{n}}} \bmod q$.
Writing any exponent $k \in [1, n]$ as $k = a + d\ceil{\sqrt{n}}$ with $a \in [0, \ceil{\sqrt{n}})$ and $d \in [0, \ceil{\sqrt{n}}]$, we obtain $b^k \bmod q = (X[a] \cdot Y[d]) \bmod q$ in $\Oh(1)$ time after $\Oh(\sqrt{n})$ preprocessing.
\end{proof}

Our method is very similar to in-place fingerprinting \cite{in_place_fingerprinting}, but not in-place, as we still need fast direct access to $T$.

\paragraph{Adaptive Choice of $\Lambda$.}
We choose sampled lengths $3 \leq \lambda_1 < \cdots < \lambda_M = \lambda_{\max}$ with a budget $\theta = 2|C|$ on the total number of stored intervals; lengths $1$ and $2$ use plain arrays of size $\sigma$ and $\sigma^2$.
We set $\lambda_{\max}$ to the mean approximate phrase length $\lfloor n/z'' \rfloor$, capped at the largest $\LCP_C$ entry.
For a pattern length $\lambda$, let
\[
  \chi(\lambda) = \bigl|\{i \in [1,|C|] : \LCP_C[i] < \lambda\}\bigr| .
\]
Each distinct length-$\lambda$ interval starts at an $\LCP_C$ entry below $\lambda$, so $|H(\SIV^\lambda_C)| \leq \chi(\lambda)$.
Sorting $\LCP_C$ into $\LCP'_C$ lets us evaluate $\chi(\lambda)$ by binary search.

We want the number of stored intervals to grow linearly with the \emph{index} $j$, not with the length $\lambda_j$ (see Property (\textbf{iii}) in \Cref{thm:patt_lengths}).
This places more sampled lengths in ranges containing many LCP values (\Cref{fig:sampled_lengths_choice}), allowing finer-grained interval searches there.

\begin{theorem}\label{thm:patt_lengths}
      Given integers $\lambda_\mathsf{max} \geq 3$ and $\theta \geq |C|$, let $\mu = \chi(3)$ and $\nu = \chi(\lambda_\mathsf{max})$.
      We can choose a set of pattern lengths $\Lambda=\{\lambda_1,\ldots,\lambda_M\}$ with $3\leq\lambda_1<\cdots<\lambda_M=\lambda_{\max}$ such that
      \begin{enumerate}[label=(\roman*),noitemsep,topsep=2pt]
            \item[\textup{\textbf{(i)}}] $\sum_{j=1}^{M} |H(\SIV^{\lambda_j}_C)| \leq \theta$,
            \item[\textup{\textbf{(ii)}}] $\sum_{j=1}^{M} |H(\SIV^{\lambda_j}_C)| = \Theta(\theta)$, and
            \item[\textup{\textbf{(iii)}}] $|H(\SIV^{\lambda_j}_C)| - \mu = \Theta\!\bigl(\tfrac{j}{M}(\nu-\mu)\bigr)$.
      \end{enumerate}
\end{theorem}
We prove the theorem in \apx{apx:proof}.

\begin{figure}[t]
      \centering
      \includegraphics[width=0.85\linewidth]{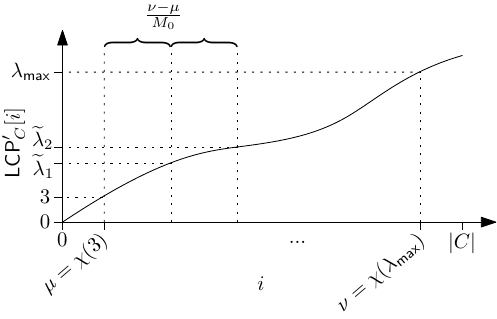}
      \caption{Quantile-based choice of candidate pattern lengths for $\Lambda$ before removing duplicates (\Cref{thm:patt_lengths}).}
      \label{fig:sampled_lengths_choice}
\end{figure}

Note that all of this can also be applied to sparse prefix array intervals, where the equality check becomes a leftward LCE, which we compute by a word-wise backward scan, and the queried patterns are additionally bounded in length by $\delta$.
For the prefix intervals, we use $\theta = 2|C|$ and $\lambda_{\max} = \delta$.

\subsection{Parallelization}\label{sec:parallel}
The construction of the SSS is already parallel \cite{lce_sss}.
We could easily parallelize every other discussed (construction) algorithm, except for the rolling hash index from \Cref{sec:hash_index}.
In the following, $p$ denotes the number of threads.

\paragraph{Gap Factorization.}
As the state of $H$ is bound to the current position in the gaps, we cannot simply split the gap regions among the threads.
Instead, we divide $T$ into $J = \lceil g / B \rceil$ blocks $T[b_0, b_1), \ldots, T[b_{J-1}, b_J)$ with $b_0 = 1$ and $b_J = n+1$, where each block spans exactly $B = \max(4096, n / (K \cdot p))$ gap positions, the last one at most $B$, and $K = 512 > 1$ is a parameter.
We iteratively factorize $p$ consecutive blocks in parallel, running through at most $K$ iterations.

Suppose we implemented this algorithm using the rolling hash index described in \Cref{sec:hash_index}.
Let $i$ be the current position in $T$ of thread $x$, and let $j$ be the current position in $T$ of thread $y > x$.
Since $y > x$ implies $j > i$, thread $y$ writes values $> i$ into $H$; thus, they cannot be used by thread $x$ to find a phrase starting at $i$.
This reduces the compression ratio, especially in the leftmost blocks per iteration.

Instead of storing $H[1..h]$, we set $h \gets h/2$ and store $H_\mathsf{read}[1..h]$ and $H_\mathsf{write}[1..h]$, and read only from $H_\mathsf{read}$ and write only to $H_\mathsf{write}$.
After each iteration, we overwrite $H_\mathsf{read}$ with $H_\mathsf{write}$.

This mitigates the issue described above, but still reduces the compression ratio, especially in the rightmost blocks per iteration, because the positions stored in $H_\mathsf{read}$ have a greater distance to those blocks.

Since we run through at most $K$ iterations, overwriting $H_\mathsf{read}$ with $H_\mathsf{write}$ takes overall $\Theta(K \cdot h / p) = \Oh(K \cdot n / p)$ time, because $h = \Oh(n)$.
Compared with the sequential algorithm, this parallel algorithm reduces the compression ratio by a factor of up to 2, especially if $0.1 < g / n < 0.5$.
We therefore fall back to the sequential algorithm whenever $g/n \leq 0.2$, where the gap factorization accounts for only a small fraction of the running time.

\paragraph{Exact algorithm.}
All data structures are built in parallel.
The factorization itself is parallelized by splitting $T$ into $16p$ contiguous regions such that each region contains the same number of approximate phrase end positions, which we then distribute over the $p$ threads dynamically, one region at a time, so that a thread that finishes its region early continues with the next one.
Then, each thread factorizes its region with perfect phrases, of which only the last phrase can be non-perfect; hence, the resulting factorization has at most $z + 16p$ phrases.

\section{Experimental Evaluation}\label{sec:experiments}

We implemented the algorithms described in \Cref{sec:algorithmic_optimizations} in \texttt{C++20}.
We measured peak memory consumption with \texttt{malloc\_count} for the factorization experiments.
For the compression and decompression experiments, where we invoke external binaries, we report the peak resident set size.
For sorting, we used the in-place sample sort implementation \texttt{ips4o} \cite{ips4o}.
Links to all software used can be found in our GitHub repository (\url{https://github.com/LukasNalbach/lz77-sss}).
The machine has two AMD EPYC 7452 CPUs ($2 \times 32$ cores at $2.35$--$3.35$\,GHz) and $1$\,TB of DDR4 RAM; we compiled with GCC 13.3.0 and \texttt{"-march=native -DNDEBUG -Ofast"} on Ubuntu 24.04.3.
All runs were pinned to the $32$ physical cores of a \emph{single} CPU, so that no measurement crosses the socket boundary and the reported scaling from $1$ to $32$ threads is free of NUMA effects.
None of the parameters of our implementation is tuned per text: we fixed each value once on a tuning corpus of eight smaller texts \cite{master_thesis} with $g/n$ from $0.0003$ to $0.70$ and used it unchanged for all measurements (see \apx{apx:parameters}).
\Apx{apx:numbers} shows the main throughput and memory numbers underlying \Cref{fig:lz,fig:compress-decompress}.

\Cref{tab:texts} shows the tested texts.
The text dewiki is highly repetitive and has been handcrafted from German Wikipedia entries.
The text chr19 consists of concatenated human chromosome 19 haplotypes, and sars2 is a collection of SARS-CoV-2 genomes, both of which were crafted out of datasets from the National Center for Biotechnology Information (NCBI) database.
Links to all texts can be found in our GitHub repository.

\begin{table}[t]
\centering
\begin{tabular}{l|r|r|r|r}
\hline
text & text size & $\sigma$ & $g/n$ & $n/z$ \\
\hline
sars2  & 50.00 GB & 80 & 0.0102 & 2760.2 \\
chr19  & 50.00 GB & 52 & 0.0247 & 1375.4 \\
dewiki & 50.00 GB & 209 & 0.0163 & 1618.2 \\
\hline
\end{tabular}
\caption{Statistics of the tested texts.}
\label{tab:texts}
\end{table}

\begin{figure*}[t]

\hspace*{-5pt}%
\begin{tikzpicture}[marks]
\begin{groupplot}[
    group style={group size=3 by 1, horizontal sep=22pt},
    throughput_style, xmode=log, ymode=log,
    width=140pt, height=130.5pt,
    xtick={0.01,0.1,1,10,100}, ytick={0.01,0.1,1,10,100,1000},
]

\nextgroupplot[
    title={sars2},
    xmin=0.6454, xmax=30.94, ymin=0.3735, ymax=391.9,
    ylabel={throughput [MB/s]},
    xlabel={peak memory consumption [bytes/$n$]},
]

\addplot[3_aprx] coordinates { (1.17445,30.9947) (1.17444,196.409) };

\addplot[exact] coordinates { (1.31193,0.745169) (1.31193,15.6224) };

\addplot[exact_smpl] coordinates { (1.77174,1.23195) (1.80272,31.827) };

\addplot[lpf] coordinates { (17,2.54177) (17.0004,14.4456) };

\addplot[kkp2] coordinates { (17,2.03275) };

\annotatenear{1.17445,30.9947, -12pt, 9pt, 1.40}
\annotatenear{1.17444,196.409, 15pt, 0pt, 1.40}
\annotatenear{1.31193,0.745169, -12pt, 9pt, 1.00}
\annotatenear{1.31193,15.6224, -15pt, 0pt, 1.00}
\annotatenear{1.77174,1.23195, 6pt, -10pt, 1.00}
\annotatenear{1.80272,31.827, 6pt, 10pt, 1.00}
\annotatenear{17,2.54177, 12pt, 9pt, 1.00}
\annotatenear{17.0004,14.4456, 6pt, 10pt, 1.00}
\annotatenear{17,2.03275, -9pt, -10pt, 1.00}

\nextgroupplot[
    title={chr19},
    xmin=0.655, xmax=30.94, ymin=0.3935, ymax=186.8,
    xlabel={peak memory consumption [bytes/$n$]},
]

\addplot[3_aprx] coordinates { (1.19183,33.3941) (1.19183,93.6366) };

\addplot[exact] coordinates { (1.37517,0.785167) (1.37658,16.0962) };

\addplot[exact_smpl] coordinates { (1.87719,2.21895) (1.87864,28.9618) };

\addplot[lpf] coordinates { (17,2.02384) (17.0004,20.5907) };

\addplot[kkp2] coordinates { (17,2.06581) };

\annotatenear{1.19183,33.3941, -12pt, 9pt, 2.02}
\annotatenear{1.19183,93.6366, 6pt, 10pt, 2.02}
\annotatenear{1.37517,0.785167, -6pt, -10pt, 1.00}
\annotatenear{1.37658,16.0962, -15pt, 2pt, 1.00}
\annotatenear{1.87719,2.21895, 0pt, -11pt, 1.00}
\annotatenear{1.87864,28.9618, 6pt, 10pt, 1.00}
\annotatenear{17,2.02384, 12pt, 9pt, 1.00}
\annotatenear{17.0004,20.5907, 6pt, 10pt, 1.00}
\annotatenear{17,2.06581, -4pt, -10pt, 1.00}

\nextgroupplot[
    title={dewiki},
    xmin=0.6485, xmax=30.94, ymin=0.7097, ymax=270.9,
    xlabel={peak memory consumption [bytes/$n$]},
]

\addplot[3_aprx] coordinates { (1.18002,43.6005) (1.18002,135.782) };

\addplot[exact] coordinates { (1.33332,1.41595) (1.3362,32.2297) };

\addplot[exact_smpl] coordinates { (1.81917,4.47431) (1.8221,58.8246) };

\addplot[lpf] coordinates { (17,2.29106) (17.0004,20.1883) };

\addplot[kkp2] coordinates { (17,2.43683) };

\annotatenear{1.18002,43.6005, -12pt, 9pt, 1.44}
\annotatenear{1.18002,135.782, 6pt, 10pt, 1.44}
\annotatenear{1.33332,1.41595, -6pt, -10pt, 1.00}
\annotatenear{1.3362,32.2297, -12pt, -9pt, 1.00}
\annotatenear{1.81917,4.47431, 1pt, -11pt, 1.00}
\annotatenear{1.8221,58.8246, 6pt, 10pt, 1.00}
\annotatenear{17,2.29106, -14pt, 5pt, 1.00}
\annotatenear{17.0004,20.1883, 6pt, 10pt, 1.00}
\annotatenear{17,2.43683, -2pt, -11pt, 1.00}

\end{groupplot}
\end{tikzpicture}

\centering
\vspace*{0.15cm}
\begin{tikzpicture}[legendmarks]
\begin{axis}[legend,legend columns=5]

\addplot[3_aprx] coordinates { (0,0) };
\addlegendentry{\texttt{sss-3aprx}};
\addplot[exact] coordinates { (0,0) };
\addlegendentry{\texttt{sss-exact}};
\addplot[exact_smpl] coordinates { (0,0) };
\addlegendentry{\texttt{sss-exact-smpl}};
\addplot[lpf] coordinates { (0,0) };
\addlegendentry{\texttt{LPF}};
\addplot[kkp2] coordinates { (0,0) };
\addlegendentry{\texttt{KKP2}};

\end{axis}
\end{tikzpicture}
\vspace*{-0.1cm}
\caption{
LZ77 factorization throughput versus peak memory consumption.
Both axes are logarithmic.
Each point is annotated with the approximation ratio $z_{\mathrm{alg}}/z$, i.e., the number of phrases produced divided by the optimum $z$.
A line connects the $1$-thread (lower) and the $32$-thread runs (upper) of the same algorithm.
\texttt{KKP2} is sequential. \Apx{tab:numbers-lz} gives the same data numerically.
}
\label{fig:lz}
\end{figure*}
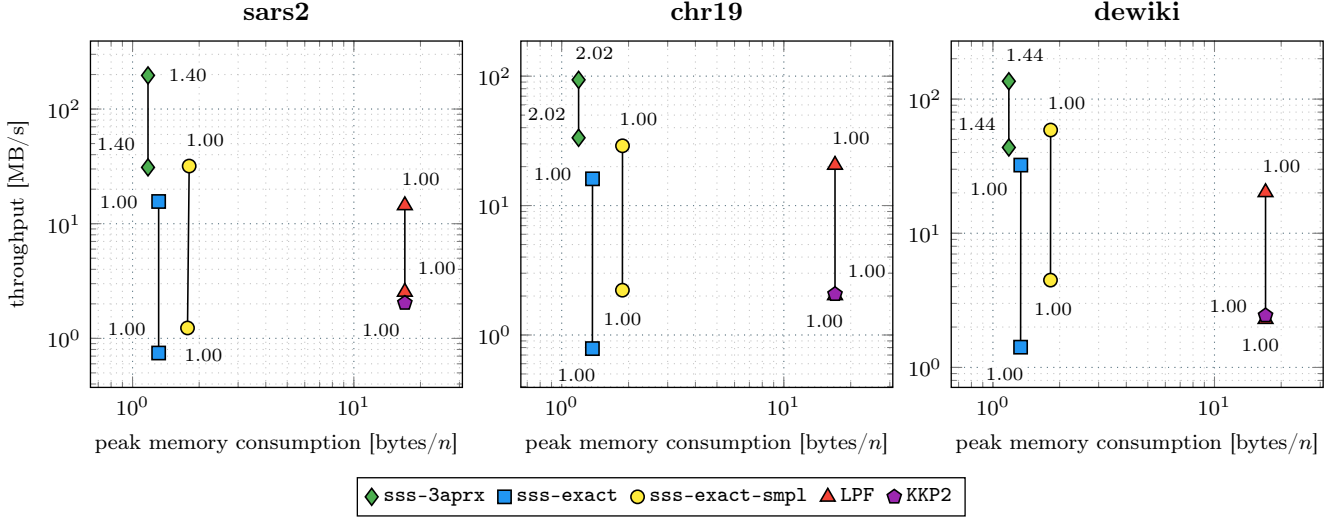

\subsection{LZ Factorization Performance}
We compare our LZ77 algorithms against the following competitors from related work:
\begin{enumerate}
      \item \texttt{sss-3aprx}: Our optimized implementation of the LZ77 3-approximation from \cite{lz77_sublinear}, described in \Cref{sec:3_apprx_opt}.
      \item \texttt{sss-exact}: Our optimized implementation of the exact LZ77 algorithm from \cite{lz77_sublinear}, described in \Cref{sec:exact_opt}, without sparse interval samples.
      \item \texttt{sss-exact-smpl}: The same as \texttt{sss-exact}, but with the sparse interval samples described in \Cref{sec:sampling}.
      \item \texttt{LPF}: The $\Oh(n)$-time LZ77 LPF algorithm (\Cref{sec:prelim}), computing $\PSV$ and $\NSV$ naively on demand; we parallelize it per region of $T$ and build $\SA$ with the multithreaded \texttt{libsais}.
      \item \texttt{KKP2}: The $\Oh(n)$-time LZ77 algorithm of K\"arkk\"ainen et al.\ \cite{kkp}; it overwrites $\SA$ in place and thus stores only two length-$n$ arrays, but is inherently sequential.
\end{enumerate}

\paragraph{Results.}
\Cref{fig:lz} shows throughput versus peak memory usage versus approximation ratio.
\texttt{LPF} and \texttt{KKP2} store two integer arrays of length $n$ on top of the text, so their peak memory is $17n$ bytes on every text.
\texttt{sss-3aprx} stores only the LCE data structure, $\NSV_S$, $\PSV_S$, and the rolling hash index (\Cref{sec:hash_index}); it uses $1.17$--$1.19n$ bytes ($14\times$ less than \texttt{LPF}), while \texttt{sss-exact} and \texttt{sss-exact-smpl}, which require more data structures, still use only $1.31$--$1.38n$ and $1.77$--$1.88n$ bytes ($12$--$13\times$ and $9\times$ less).
The approximation ratio is $1.40$--$2.02$, far below the guaranteed $3$.

On one thread, \texttt{sss-3aprx} factorizes $12$--$19\times$ faster than \texttt{LPF}, whereas \texttt{sss-exact} is $1.6$--$3.4\times$ slower than \texttt{LPF}; interval sampling makes \texttt{sss-exact-smpl} $1.7$--$3.2\times$ faster than \texttt{sss-exact}.
From 1 to 32 threads, \texttt{sss-3aprx} speeds up by $2.8$--$6.3\times$, the exact algorithms far more ($20.5$--$22.8\times$ for \texttt{sss-exact}, $13$--$26\times$ for \texttt{sss-exact-smpl}), and \texttt{LPF} itself by only $5.7$--$10.2\times$.
On 32 threads, \texttt{sss-3aprx} is therefore $4.5$--$13.6\times$ faster than \texttt{LPF}, and even \texttt{sss-exact-smpl} factorizes $1.4$--$2.9\times$ faster than the parallel \texttt{LPF} (\texttt{sss-exact} stays within $0.8$--$1.6\times$ of it).

\subsection{Compression Performance}\label{sec:compression}
We compare \ssszip against the following compressors:
\begin{enumerate}
      \item General-purpose compressors \texttt{gzip}, \texttt{bzip2}, \texttt{lz4}, \texttt{zstd}, \texttt{xz} and \texttt{7z}, each limited to the redundancy within a bounded window, block or dictionary.
      \item \texttt{bsc}: A multithreaded compressor that uses Lempel--Ziv preprocessing, block-wise Burrows--Wheeler transform (BWT), move-to-front encoding and quantized local frequency coding.
      \item \texttt{ssszip-bsc}: Our precompressor \ssszip, followed by \texttt{bsc}.
      Rather than closing the gaps of the 3-approximation, \ssszip keeps only its LPF phrases and passes them, encoded as variable-byte codes, together with the gap strings, to a configurable downstream compressor (here \texttt{bsc}).
      \item \texttt{alz-bsc}: Dinklage's approximate LZ77 precompressor \texttt{alz} \cite{alz}, followed by \texttt{bsc}.
      We use the sampling parameter $s = 2^6$, which \cite{alz} recommends as the best time/space trade-off.
      Like \ssszip, it computes long phrases that do not cover all of $T$ and leaves the remaining gaps to the downstream compressor, but computes them at the granularity of a prefix-free parsing \cite{pfp} of $T$.
\end{enumerate}
For \texttt{bsc}, both standalone and as the postcompressor of \texttt{ssszip-bsc} and \texttt{alz-bsc}, we use the maximum block size of 2047\,MB and encoding mode \texttt{"-e2"}.
\paragraph{Omitted LZ parsers.}
Beyond the compressors above, we do not separately measure the following LZ-like parsers, all of which are dominated by a competitor we do measure.
\begin{itemize}[noitemsep,topsep=2pt]
      \item \texttt{pfp-lz77} \cite{lz77_pfp}, \texttt{rle-lz77-o} \cite{lz77_online_rlbwt}, \texttt{relz} \cite{rlz_lz} and \texttt{topk-lz77} \cite{topk_lz77} are the four competitors of \texttt{alz} \cite{alz}, which outperforms each of them in running time on every tested input, by an order of magnitude in the case of \texttt{pfp-lz77}.
      \item \texttt{SE-KKP} \cite{KKP14em}, the semi-external variant of \texttt{KKP}, and \texttt{BGone} \cite{GotoBannai14}, which computes exact LZ77 in linear time from a single integer array, are both benchmarked against \texttt{pfp-lz77} by Hong and Boucher \cite{HongBoucher25}: \texttt{SE-KKP} is substantially slower, and \texttt{BGone} is of the same order (buying roughly half the total space at about twice the running time).
      Since \texttt{alz} is already an order of magnitude faster than \texttt{pfp-lz77}, both are far slower than \texttt{alz-bsc} and \texttt{ssszip-bsc}.
      Note also that neither is a precompressor, so a like-for-like comparison would additionally have to charge them the encoding step.
\end{itemize}
Since \ssszip and the 3-approximation it builds on achieve almost the same throughput and peak memory on repetitive inputs \cite{master_thesis}, these comparisons carry over to the factorization in \Cref{fig:lz}.

\begin{figure*}[t]
\input{charts/compress-decompress.tex}
\end{figure*}

\paragraph{Results.}
\Cref{fig:compress-decompress} shows compression and decompression throughput versus peak memory usage; each point is annotated with the compression ratio.

\texttt{ssszip-bsc} reaches a compression ratio $1.3$--$94\times$ that of the best of the six general-purpose compressors (\texttt{gzip} through \texttt{7z}) per text; the factor is largest on chr19, whose long-range repetitions no window, block or dictionary catches.
It attains this ratio at $2.8$--$27\times$ the compression throughput of \texttt{xz} and $3.3$--$30\times$ that of \texttt{7z} (the general-purpose compressors with the highest ratios), and at a peak memory of $1.17$--$1.19n$ bytes.
From 1 to 32 threads, \texttt{xz} loses up to $36\%$ and \texttt{7z} up to $16\%$ of its ratio, and both use more memory, because they parallelize by compressing independent chunks; the ratios of \texttt{ssszip-bsc}, \texttt{alz-bsc} and \texttt{bsc} instead stay unchanged, as does the memory of \texttt{ssszip-bsc} and \texttt{bsc}, whereas that of \texttt{alz-bsc} grows by $28\%$ on chr19 (due to the parallel merging of fingerprint hash sets).
Meanwhile, \texttt{ssszip-bsc}'s throughput grows by $4.5$--$6\times$.
\texttt{bsc}'s ratio is only $6\%$ below that of \texttt{ssszip-bsc} on dewiki, but $28\times$ lower on chr19, where a single block no longer captures the long-range repetitions.
\texttt{alz-bsc} achieves ratios similar to \texttt{ssszip-bsc}: $19\%$ and $5\%$ higher on sars2 and dewiki, and $2.5\%$ lower on chr19, where its phrases are bounded by the granularity of the prefix-free parsing.
It is, however, consistently more expensive: it needs $1.48$--$1.49n$ bytes ($1.25\times$ the memory of \texttt{ssszip-bsc}), compresses $1.2$--$1.5\times$ slower on one thread, and does not scale as well, so that on 32 threads \texttt{ssszip-bsc} is $1.9\times$ (sars2), $3.9\times$ (chr19) and $1.7\times$ (dewiki) faster.
Its memory grows as the sampling parameter $s$ decreases; \cite{alz} reports it to fall below the input size only around $s = 2^7$, so the peak we observe at the recommended $s = 2^6$ is expected.

During decompression, \texttt{ssszip-bsc} makes a single pass over the phrases and needs $0.06$--$0.12n$ bytes, whereas \texttt{bsc} must invert the BWT of an entire block and uses $15$--$50\times$ as much ($0.97$--$4.49n$ bytes); \texttt{bsc} is also $8\times$ (sars2) and $16\times$ (chr19) slower than \texttt{ssszip-bsc}, and faster only on dewiki.

\section{Conclusion}\label{sec:conclusion}
We gave the first practical implementation of the small-space LZ77 algorithms from \cite{lz77_sublinear}, a $3$-approximation and an exact algorithm, by replacing their two components that resist a direct implementation -- a short-pattern lookup table that degenerates for realistic inputs and an orthogonal range reporting data structure -- and tuning every remaining stage.
This targets LZ-based compressed indexing, where the number of phrases $z$, rather than the encoded byte size, governs the index size, and where computing a factorization with few phrases in little space is a construction bottleneck.
On 32 threads, our 3-approximation factorizes $4.5$--$13.6\times$ faster than the classical linear-space algorithm at $14\times$ less memory, and even our exact algorithm is $1.4$--$2.9\times$ faster at $9\times$ less memory, while staying far below the guaranteed $3z$ phrases in practice.
As a side result, feeding only the perfect phrases to a general-purpose compressor yields \ssszip, a precompressor on par with the state-of-the-art \texttt{alz} \cite{alz} in compression ratio while using less memory and scaling better in parallel.

\clearpage
\bibliographystyle{siamplain}
\bibliography{paper}

\iffullversion
\appendix

\crefalias{section}{appendix}

\section{Proof of Theorem~\ref{thm:patt_lengths}}\label{apx:proof}
\begin{restatedtheorem}[\Cref{thm:patt_lengths}]
      Given integers $\lambda_\mathsf{max} \geq 3$ and $\theta \geq |C|$, let $\mu = \chi(3)$ and $\nu = \chi(\lambda_\mathsf{max})$.
      We can choose a set of pattern lengths $\Lambda=\{\lambda_1,\ldots,\lambda_M\}$ with $3\leq\lambda_1<\cdots<\lambda_M=\lambda_{\max}$ such that
      \begin{enumerate}[label=(\roman*),noitemsep,topsep=2pt]
            \item[\textup{\textbf{(i)}}] $\sum_{j=1}^{M} |H(\SIV^{\lambda_j}_C)| \leq \theta$,
            \item[\textup{\textbf{(ii)}}] $\sum_{j=1}^{M} |H(\SIV^{\lambda_j}_C)| = \Theta(\theta)$, and
            \item[\textup{\textbf{(iii)}}] $|H(\SIV^{\lambda_j}_C)| - \mu = \Theta\!\bigl(\tfrac{j}{M}(\nu-\mu)\bigr)$.
      \end{enumerate}
\end{restatedtheorem}
\begin{proof}
Let $\mu = \chi(3)$ and $\nu = \chi(\lambda_{\max})$.
If $\nu=\mu$, choose $\Lambda=\{\lambda_{\max}\}$; then $|H(\SIV^{\lambda_{\max}}_C)|\leq\nu\leq\theta$ gives (\textbf{i}).
For $\nu>\mu$, let $M_0 \geq 1$ be the number of candidate lengths.
Choose $\widetilde{\lambda}_{M_0}=\lambda_{\max}$ and, for $j \in [1,M_0)$,
\[
      \widetilde{\lambda}_j
      = \LCP'_C\!\left[\left\lfloor\mu+\frac{j}{M_0}(\nu-\mu)\right\rfloor+1\right]
\]
(see \Cref{fig:sampled_lengths_choice}).
All candidates lie in $[3,\lambda_{\max}]$, and, for every $j \in [1,M_0]$,
\begin{equation}\label{eq:Hsize}
      |H(\SIV^{\widetilde{\lambda}_j}_C)|
      \leq \chi(\widetilde{\lambda}_j)
      \leq \mu+\frac{j}{M_0}(\nu-\mu).
\end{equation}
We have not yet fixed $M_0$.
Summing \eqref{eq:Hsize} gives
\begin{align*}
      \sum_{j=1}^{M_0}|H(\SIV^{\widetilde{\lambda}_j}_C)|
      &\leq M_0\mu+\frac{\nu-\mu}{M_0}\sum_{j=1}^{M_0}j \\
      &= M_0\,\frac{\mu+\nu}{2}+\frac{\nu-\mu}{2}.
\end{align*}
This upper bound is at most $\theta$ if and only if $M_0\leq(2\theta+\mu-\nu)/(\mu+\nu)$.
Thus, choose $M_0=\lfloor(2\theta+\mu-\nu)/(\mu+\nu)\rfloor$.
We have $M_0\geq1$ because $\theta\geq|C|\geq\nu$.
Removing duplicate lengths gives $\Lambda$ and can only decrease the sum, proving (\textbf{i}).

For (\textbf{ii}) and (\textbf{iii}), assume that the candidates are distinct (so $M=M_0$) and that, for every $j\in[1,M]$,
\begin{align*}
      \text{(\textbf{iv})}\quad
      \chi(\lambda_j)-\mu
      &=\Theta\!\left(\frac{j}{M}(\nu-\mu)\right), \\
      \text{(\textbf{v})}\quad
      |\{p\in C:n-p+1<\lambda_j\}|
      &\leq\tfrac12\bigl(\chi(\lambda_j)-\mu\bigr).
\end{align*}
Assumption (\textbf{iv}) means that the number of LCP values below $\lambda_j$ grows linearly with $j$, starting from $\mu$.
Some sampled suffixes may be shorter than $\lambda_j$; each contributes one to $\chi(\lambda_j)$ but no length-$\lambda_j$ interval.
Assumption (\textbf{v}) ensures that subtracting these suffixes preserves at least half of the increase above $\mu$.
Hence,
\begin{equation}\label{eq:linear_sizes}
      \tfrac12\bigl(\chi(\lambda_j)-\mu\bigr)
      \leq |H(\SIV^{\lambda_j}_C)|-\mu
      \leq \chi(\lambda_j)-\mu.
\end{equation}
Together with (\textbf{iv}), this proves (\textbf{iii}).
For (\textbf{ii}), sum this estimate as in the budget calculation:
\begin{align*}
      \sum_{j=1}^{M}|H(\SIV^{\lambda_j}_C)|
      &=\Theta\!\left(M\mu+\frac{\nu-\mu}{M}\sum_{j=1}^{M}j\right) \\
      &=\Theta\!\left(M\,\frac{\mu+\nu}{2}+\frac{\nu-\mu}{2}\right)=\Theta(\theta).
\end{align*}
Our choice of $M=M_0$ puts the expression inside $\Theta$ in $(\theta/2,\theta]$.
\end{proof}

\section{Parameter Choice and Tuning}\label{apx:parameters}
Our implementation has eight tunable parameters, whose values the main text states where they are introduced; this appendix explains how we chose them.

\subsection{Tuning Methodology}\label{apx:methodology}
The tuning corpus is the one of \cite{master_thesis} and is disjoint from the texts of \Cref{tab:texts} in size and, for four of the eight texts, also in content: \texttt{einstein.en.txt}, \texttt{cere}, \texttt{english} and \texttt{boost} from the Pizza\&Chili corpus, plus $20$\,GiB prefixes of dewiki, sars2 and chr19 and a $16$\,GiB Common Crawl text.
It deliberately spans the full repetitiveness range, from $g/n = 0.0003$ (\texttt{boost}) to $g/n = 0.70$ (\texttt{english}), so that a parameter that only works on highly repetitive inputs is visibly penalized.
We tuned one parameter at a time, sweeping it over a range of powers of two while holding the others at their current value, and selected the value that was best (or tied for best) on the majority of the corpus with respect to throughput, peak memory and, where applicable, approximation ratio.
Because the tuning corpus is much smaller than the $50$\,GB texts of \Cref{sec:experiments}, all parameters are expressed relative to $n$, $z$, $|C|$ or $g$ rather than as absolute sizes, so that they transfer to inputs of a different scale.

\subsection{\texorpdfstring{SSS Parameter $\tau$}{SSS Parameter tau}}\label{apx:tau}
Recall that $S$ is a $\tau$-synchronizing set (\Cref{def:sss}) and that $|S| = \Oh(n/\tau)$.
Sweeping $\tau = 2^i$ for $i \in [4,11]$ leaves the approximation ratio essentially unchanged, because it does not affect which phrases are perfect.
It does trade memory against time: $\tau \in \{128,256\}$ can be faster on texts with long periodic regions (\texttt{cere}), but $\tau \geq 512$ both lowers the peak memory (the SSS shrinks as $\tau$ grows) and raises the throughput on all non-periodic texts, which is why we settled on $\tau = 512$ \cite[Section 5.1]{master_thesis}.

\subsection{\texorpdfstring{Rolling Hash Index Size $h$}{Rolling Hash Index Size h}}\label{apx:hash_index_size}
Recall from \Cref{sec:hash_index} that $h$ is the number of entries of the rolling hash index $H$: we give $H$ a memory budget of $\MH = \min(2^{30}, \max(2^{20}, n/10, g/3, \Mpeak - \Mcur))$ bytes and set $h$ to the power of two that results in the size of $H$ being closest to $\MH$.
On unrepetitive texts, $g \approx n$, so $g/3$ gives $H$ about $n/3$ bytes.
On repetitive texts, where $g \approx 0$ and $g/3$ alone would make $H$ so small that collisions dominate, $n/10$ takes over; it is the smallest size for which the approximation ratio no longer improved.
The term $\Mpeak - \Mcur$ is the memory that the LPF phase has allocated and already released, and we let $H$ occupy it, because this does not increase the peak memory consumption.
Finally, we clamp $\MH$ to $[2^{20}, 2^{30}]$.

\subsection{\texorpdfstring{Pattern Lengths $L$ of the Rolling Hash Index}{Pattern Lengths L of the Rolling Hash Index}}\label{apx:patt_lens}
Recall that $L$ holds the pattern lengths whose rolling fingerprints we maintain while scanning the gaps, and that we look up each of them in $H$ at every gap position (\Cref{sec:hash_index}).
$L$ (\Cref{tab:lengths}) holds five lengths, which was the smallest number at which adding a sixth no longer improved the approximation ratio while each additional length costs one rolling fingerprint per text position.
The five values are chosen from the gap phrase length guess $\rho = \min(\bar{g}, \bar{l}, 2^{10-7g/n})$: the shortest length is always $2$ so that even very short gap phrases are found, the longest tracks $\rho$, and the remaining three interpolate roughly geometrically.
The third term of $\rho$ keeps the schedule from drifting to long patterns on unrepetitive texts, where gap phrases are short and long patterns would simply never match; it decays from $2^{10}$ at $g/n = 0$ to $8$ at $g/n = 1$.
Since the rolling fingerprints are never stored, space is no concern here, and we use the large Mersenne prime $q = 2^{107}-1$ with 128-bit registers, which minimizes collisions and thereby maximizes the compression ratio.

\subsection{\texorpdfstring{Grid Cell Width $D$}{Grid Cell Width D}}\label{apx:grid}
Recall that the grid of \Cref{sec:grid} divides the point range $[1,N]^2$ into cells of side length $D$, and that $N = |C|$ here.
A \swOroR query costs $\Oh(|Q|/D^2 + N/D + D)$ time, so $D$ trades the cost of the fully contained cells against that of the boundary cells, while the cell array itself shrinks as $D$ grows.
Sweeping $D = 2^i$, we found $D = 2^{14}$ to give the best query throughput; at that point the array $I$ occupies a negligible fraction of the data structure, so larger $D$ buys no space worth the extra boundary work.
The character decomposition of \Cref{thm:decomp} shrinks the effective grid area by up to a factor of $\sigma$, which makes a small $D$ affordable.

\subsection{\texorpdfstring{Small-Range Threshold $\gamma$}{Small-Range Threshold gamma}}\label{apx:gamma}
Recall that $\gamma$ is the width below which we call a query rectangle narrow, i.e.,\ $\min(x_2-x_1, y_2-y_1) \leq \gamma$ (\Cref{sec:oror}).
For such a rectangle, we answer the query with a linear scan over $\Pi$ and $\Psi$ instead of querying the grid, which costs one random memory access instead of several.
$\gamma = 4096$ is the crossover we measured between the two: the scan touches $\gamma$ consecutive entries, which is still cheaper than the grid's random accesses, and the queries of the exact algorithm are narrow in one dimension often enough for this to matter.

\subsection{\texorpdfstring{Fingerprint Sampling Rate $s$}{Fingerprint Sampling Rate s}}\label{apx:fp_rate}
Recall that we store the Karp--Rabin fingerprint of every $s$th prefix of $T$ and obtain the fingerprint of an arbitrary substring from two of them (\Cref{thm:fingerprints}).
\Cref{thm:fingerprints} gives an $\Oh(n/s + \sqrt n)$-size structure with $\Oh(s)$ query time, so $s$ is a direct time/space dial.
Sweeping $s = 2^i$, we found $s = 16$ together with $32$-bit fingerprints ($q = 2^{31}-1$) to be the best trade-off: the array then costs $n/4$ bytes, while the $\Oh(s)$ extension loop stays within a few cache lines.
We can use a $32$-bit prime here, because the fingerprints only serve as the hash function of the interval tables and every lookup verifies its result with an LCE query (\Cref{sec:sampling}).

\subsection{\texorpdfstring{Interval Sampling Budget $\theta$ and $\lambda_{\max}$}{Interval Sampling Budget theta and lambda max}}\label{apx:sampling}
Recall that $\Lambda \subseteq [3, \lambda_{\max}]$ holds the pattern lengths for which we precompute interval tables, so $\lambda_{\max}$ is the longest pattern we sample (\Cref{sec:sampling}).
$\theta$ bounds the total number of stored intervals, so it is the memory dial of \Cref{sec:sampling}; $\theta = 2|C|$ was the largest value in our sweep for which the added memory was still repaid by the reduced running time.
For the suffix intervals we set $\lambda_{\max} = \lfloor n/z'' \rfloor$, using the mean approximate phrase length as a proxy for the lengths of the computed exact phrases.
We also cap $\lambda_{\max}$ at the largest $\LCP_C$ entry, because every interval of a greater length is a singleton.
For the sparse \emph{prefix} array intervals we keep $\theta = 2|C|$ but use $\lambda_{\max} = \delta = \min(\lfloor n/z'' \rfloor, 256)$, because there we only ever query patterns of length at most $\delta$, and allowing $\delta > 256$ neither improved throughput nor reduced peak memory.

\subsection{\texorpdfstring{Parallel Block Factor $K$}{Parallel Block Factor K}}\label{apx:parallel}
Recall that each block of the parallel gap factorization holds $B = \max(4096, n/(K \cdot p))$ gap positions (\Cref{sec:parallel}).
$K$ bounds the number of iterations, and overwriting $\Hread$ with $\Hwrite$ costs $\Oh(K \cdot n / p)$ time in total, so large $K$ costs time while small $K$ costs compression ratio (larger blocks result in more distant text positions in $\Hread$).
$K = 512$ was the smallest value in our sweep at which the ratio had converged.
For the same reason we fall back to the sequential gap factorization when $g/n \leq 0.2$: there, the gap factorization is a small fraction of the running time, so parallelizing it would only cost compression ratio.

\section{Detailed Measurements}\label{apx:numbers}
\Cref{tab:numbers-lz,tab:numbers-zip} list the numbers behind \Cref{fig:lz,fig:compress-decompress}.

\renewcommand{\dbltopfraction}{0.95}
\begin{table*}[t]
\centering
\setlength{\tabcolsep}{3pt}
\begin{tabular}{l|rrrrr|rrrrr|rrrrr}
\hline
 & \multicolumn{5}{c|}{sars2} & \multicolumn{5}{c|}{chr19} & \multicolumn{5}{c}{dewiki} \\
 & \texttt{3aprx} & \texttt{exact} & \texttt{smpl} & \texttt{LPF} & \texttt{KKP2} & \texttt{3aprx} & \texttt{exact} & \texttt{smpl} & \texttt{LPF} & \texttt{KKP2} & \texttt{3aprx} & \texttt{exact} & \texttt{smpl} & \texttt{LPF} & \texttt{KKP2} \\
\hline
thr.\ $p{=}1$ & 31.0 & 0.7 & 1.2 & 2.5 & 2.0 & 33.4 & 0.8 & 2.2 & 2.0 & 2.1 & 43.6 & 1.4 & 4.5 & 2.3 & 2.4 \\
thr.\ $p{=}32$ & 196.4 & 15.6 & 31.8 & 14.4 & -- & 93.6 & 16.1 & 29.0 & 20.6 & -- & 135.8 & 32.2 & 58.8 & 20.2 & -- \\
mem.\ $p{=}1$ & 1.174 & 1.312 & 1.772 & 17.0 & 17.0 & 1.192 & 1.375 & 1.877 & 17.0 & 17.0 & 1.180 & 1.333 & 1.819 & 17.0 & 17.0 \\
mem.\ $p{=}32$ & 1.174 & 1.312 & 1.803 & 17.0 & -- & 1.192 & 1.377 & 1.879 & 17.0 & -- & 1.180 & 1.336 & 1.822 & 17.0 & -- \\
$z_\mathsf{alg}/z$ & 1.40 & 1.00 & 1.00 & 1.00 & 1.00 & 2.02 & 1.00 & 1.00 & 1.00 & 1.00 & 1.44 & 1.00 & 1.00 & 1.00 & 1.00 \\
\hline
\end{tabular}
\caption{Numbers underlying \Cref{fig:lz}: factorization throughput [MB/s], peak memory consumption [bytes/$n$, text included] and approximation ratio, on $p$ threads. Column headers abbreviate \texttt{sss-3aprx}, \texttt{sss-exact}, \texttt{sss-exact-smpl}, \texttt{LPF} and \texttt{KKP2}; \texttt{KKP2} is sequential.}
\label{tab:numbers-lz}
\end{table*}

\begin{table*}[t]
\centering
\setlength{\tabcolsep}{4pt}
\begin{tabular}{l|rrrrrrrrr}
\hline
 & \texttt{ssszip-bsc} & \texttt{alz-bsc} & \texttt{bsc} & \texttt{gzip} & \texttt{bzip2} & \texttt{lz4} & \texttt{zstd} & \texttt{xz} & \texttt{7z} \\
\hline
\multicolumn{10}{l}{\emph{sars2}} \\
\hspace{3pt}comp.\ thr.\ $p{=}1$ & 29.0 & 25.2 & 12.0 & 9.8 & 10.4 & 498.7 & 880.6 & 5.8 & 5.5 \\
\hspace{3pt}comp.\ thr.\ $p{=}32$ & 174.4 & 91.1 & 333.0 & -- & -- & -- & 2277 & 173.3 & 80.7 \\
\hspace{3pt}comp.\ mem.\ $p{=}1$ & 1.174 & 1.475 & 3.500 & $2.1{\cdot}10^{-5}$ & $1.6{\cdot}10^{-4}$ & $1.4{\cdot}10^{-4}$ & $7.4{\cdot}10^{-4}$ & $1.9{\cdot}10^{-3}$ & $3.9{\cdot}10^{-3}$ \\
\hspace{3pt}comp.\ mem.\ $p{=}32$ & 1.174 & 1.480 & 3.499 & -- & -- & -- & $6.4{\cdot}10^{-3}$ & 0.077 & 0.075 \\
\hspace{3pt}comp.\ ratio $p{=}1$ & 490 & 583 & 413 & 5.75 & 18.4 & 2.98 & 55.6 & 244 & 255 \\
\hspace{3pt}comp.\ ratio $p{=}32$ & 490 & 583 & 413 & -- & -- & -- & 55.6 & 208 & 235 \\
\hspace{3pt}decomp.\ thr. & 253.3 & 678.9 & 32.2 & 191.4 & 39.7 & 604.1 & 904.6 & 580.1 & 442.9 \\
\hspace{3pt}decomp.\ mem. & 0.062 & 0.051 & 3.113 & $2.1{\cdot}10^{-5}$ & $8.3{\cdot}10^{-5}$ & $1.4{\cdot}10^{-4}$ & $8.2{\cdot}10^{-5}$ & $2.1{\cdot}10^{-4}$ & $4.5{\cdot}10^{-4}$ \\
\hline
\multicolumn{10}{l}{\emph{chr19}} \\
\hspace{3pt}comp.\ thr.\ $p{=}1$ & 28.2 & 18.8 & 6.9 & 6.3 & 12.4 & 355.0 & 167.9 & 1.1 & 1.0 \\
\hspace{3pt}comp.\ thr.\ $p{=}32$ & 155.3 & 39.5 & 179.2 & -- & -- & -- & 1689 & 31.4 & 15.3 \\
\hspace{3pt}comp.\ mem.\ $p{=}1$ & 1.192 & 1.494 & 4.619 & $2.1{\cdot}10^{-5}$ & $1.6{\cdot}10^{-4}$ & $1.4{\cdot}10^{-4}$ & $8.4{\cdot}10^{-4}$ & $1.9{\cdot}10^{-3}$ & $3.9{\cdot}10^{-3}$ \\
\hspace{3pt}comp.\ mem.\ $p{=}32$ & 1.192 & 1.906 & 4.619 & -- & -- & -- & $7.7{\cdot}10^{-3}$ & 0.083 & 0.081 \\
\hspace{3pt}comp.\ ratio $p{=}1$ & 441 & 430 & 15.7 & 3.65 & 4.08 & 1.92 & 3.53 & 4.67 & 4.65 \\
\hspace{3pt}comp.\ ratio $p{=}32$ & 441 & 430 & 15.7 & -- & -- & -- & 3.53 & 4.64 & 4.64 \\
\hspace{3pt}decomp.\ thr. & 330.7 & 197.9 & 20.1 & 156.1 & 27.3 & 540.3 & 520.7 & 88.3 & 89.6 \\
\hspace{3pt}decomp.\ mem. & 0.122 & 0.183 & 4.494 & $2.1{\cdot}10^{-5}$ & $8.3{\cdot}10^{-5}$ & $1.4{\cdot}10^{-4}$ & $1.0{\cdot}10^{-4}$ & $2.1{\cdot}10^{-4}$ & $7.2{\cdot}10^{-4}$ \\
\hline
\multicolumn{10}{l}{\emph{dewiki}} \\
\hspace{3pt}comp.\ thr.\ $p{=}1$ & 29.8 & 23.7 & 430.1 & 23.5 & 11.0 & 400.3 & 1118 & 10.7 & 8.9 \\
\hspace{3pt}comp.\ thr.\ $p{=}32$ & 133.6 & 78.2 & 1147 & -- & -- & -- & 1794 & 320.3 & 158.8 \\
\hspace{3pt}comp.\ mem.\ $p{=}1$ & 1.180 & 1.486 & 0.984 & $2.1{\cdot}10^{-5}$ & $1.6{\cdot}10^{-4}$ & $1.4{\cdot}10^{-4}$ & $7.6{\cdot}10^{-4}$ & $1.9{\cdot}10^{-3}$ & $3.9{\cdot}10^{-3}$ \\
\hspace{3pt}comp.\ mem.\ $p{=}32$ & 1.180 & 1.491 & 0.984 & -- & -- & -- & $6.0{\cdot}10^{-3}$ & 0.077 & 0.075 \\
\hspace{3pt}comp.\ ratio $p{=}1$ & 490 & 515 & 461 & 3.32 & 14.9 & 2.48 & 106 & 391 & 364 \\
\hspace{3pt}comp.\ ratio $p{=}32$ & 490 & 515 & 461 & -- & -- & -- & 106 & 251 & 304 \\
\hspace{3pt}decomp.\ thr. & 354.9 & 502.2 & 476.1 & 152.5 & 41.6 & 616.7 & 918.6 & 621.0 & 455.8 \\
\hspace{3pt}decomp.\ mem. & 0.065 & 0.082 & 0.967 & $2.1{\cdot}10^{-5}$ & $8.3{\cdot}10^{-5}$ & $1.4{\cdot}10^{-4}$ & $8.2{\cdot}10^{-5}$ & $2.1{\cdot}10^{-4}$ & $4.5{\cdot}10^{-4}$ \\
\hline
\end{tabular}
\caption{Numbers underlying \Cref{fig:compress-decompress}: (de)compression throughput [MB/s], peak memory consumption [bytes/$n$] and compression ratio (uncompressed size divided by compressed size), for $p \in \{1,32\}$ threads. \texttt{gzip}, \texttt{bzip2} and \texttt{lz4} are sequential; decompression is single-threaded throughout.}
\label{tab:numbers-zip}
\end{table*}

\fi

\end{document}